%% file: main.tex
\documentclass{report}
\usepackage[utf8]{inputenc}
\usepackage{amsmath,amssymb,mathrsfs,amsthm,parskip,mathtools,booktabs}
\usepackage[sorting=none]{biblatex}
\usepackage{hyperref}
\usepackage{float}
\usepackage[a4paper, total={6in, 8in}]{geometry}

\theoremstyle{definition}
\newtheorem{dfn}{Definition}
\newtheorem{lem}{Lemma}
\newtheorem{cor}{Corollary}
\newtheorem{thm}{Theorem}

\newtheorem{ax}{Axiom}
\newtheorem{rem}{Remark}
\newtheorem*{rem*}{Remark}

\newcommand{\hodge}{{*}}
\newcommand{\bfhodge}{{\boldsymbol{*}}}
\newcommand{\dif}{\mathrm{d}}
\newcommand{\imult}{\mathbin{\lrcorner}}

\begin{document}
\begin{center} 
\par 
\bigskip 
{\bf Filippo Saatkamp\footnote{Master's student of mathematics at the LMU Munich}}\\  
\noindent
filippo.saatkamp@gmail.com
\par
\rm \normalsize
\bigskip \bigskip \bigskip 

\par \par \par \par \par \par \par \par \par 
\par 
\par 
\par 
\par 
\bigskip \Large
\noindent

{\bf\LARGE Reformulation of Special Relativity and Electromagnetism in terms of Reference Frames defined as maps from spacetime onto affine spaces\par}

\end{center}
\newpage
\section*{Abstract}
\input{Abstract}
\section*{Acknowledgements}
A heartfelt thanks goes to Valter Moretti for carefully reviewing the manuscript and giving valuable feedback. More generally, his didactic work - in particular \cite{Mor2020} - has been a fundamental inspiration and I am thankful for his many elaborate answers to my questions.
\tableofcontents

\chapter{Reference Frames}
\section{Mathematical setup}
\input{1/11-Non-negative}
\section{Reference Frames}
\input{1/12-Charts}
\section{Generalized Manifolds and Tangent Bundles}
\input{1/13-GM}

\chapter{Classical Mechanics vs. Special Relativity}
\section{Galilean Transformations}
\input{2/21-Galilean}
\section{Lorentz and Poincaré Transformations}
\input{2/22-PoinDefinition}
\section{Representation of Lorentz transformations}
\input{2/23-Rep}
\section{Orientations}\label{section_orietations}
\input{2/24-Orientations}
\section{World Lines}
\input{2/25-WorldLine}
\section{Transformation of vectors}\label{section_transformation}
\input{2/26-Vector_Transformation}
\section{Velocity Reciprocity}
\input{2/27-VR}
\section{Interpretation of boosts}
\input{2/28-Boosts}
\section{Inertial frames and accelerated frames}
\subsection*{Frames accelerated with respect to another frame}
\input{2/29/291-acc.tex}
\subsection*{Transformation of velocities and accelerations}
\input{2/29/292-trans.tex}
\subsection*{Inertial frames in Newtonian mechanics}
\input{2/29/293-inertial.tex}

\chapter{Special Relativity and Electromagnetism}
From now on we consider a set of reference frames on a set $M$ such that all transition functions are Poincaré transformations.
\section{Proper Time}
\input{3/31-Proper_Time}
\section{The Riemannian Metric}
\input{3/32-Riemann}
\section{4-vectors}
\input{3/33-Velocity}
\section{Covariant Electromagnetism}
\input{3/34-EM_Tensor}
\printbibliography
\end{document}

%% file: Abstract.tex
The starting point of this paper is one of the several definitions of reference frames (frames for short) introduced in \cite{Mor2020}: In classical mechanics a frame can be defined as a triple $(E,\Pi,t)$, where $E$ is a 3-dimensional euclidean space, $\Pi$ maps spacetime $M$ onto $E$ and $t$ maps $M$ onto $\mathbb R$. The definition allows an intuitive and coordinate-free formulation of Newtonian mechanics in terms of frames instead of coordinates \cite{Mor2020}. In particular, the postulate of a set of charts on $M$ (an atlas) is replaced by the postulate of a set of frames. Then the charts can be re-obtained through the choice of affine coordinates.

The main point of this paper is to continue the work and to reformulate special relativity and electromagnetism in terms of frames. In addition, we make some modifications to the definition of frames: In order to reflect the possibility to choose different origins of time, a frame is defined to be a quadruple - the additional item is a 1-dimensional affine space $A$ and $t$ maps $M$ onto $A$. In addition, units enter the theory as elements of positive spaces and we obtain a geometric and manifestly unit-independent reformulation.

Each frame allows us to identify spacetime with a 4-dimensional product-space and we can generalize the definition of differentiable manifolds such that the frames turn out to form an atlas. Then the only difference between Newtonian mechanics and special relativity is the assumed type of transition functions - Galilean transformations or Poincaré transformations between affine spaces. This allows us to highlight the common features and differences in the second chapter of the paper. For example, we define velocity reciprocity in mathematical terms and prove the phenomenon for both kind of transformations. The chapter concludes with a discussion of inertial and accelerated frames.

In the third chapter we restrict our attention to the reformulation of special relativity and electromagnetism with a strong emphasis on covariance. World lines are introduced as particular subsets of spacetime, the proper time of world line is defined as an affine structure on the world line and the reformulation of electromagnetism is based on the representation of differential forms w.r.t. a frame.

%% file: 1/11-Non-negative.tex
Units are elements of positive spaces - this statement simply summarizes the commonly accepted axioms for units \cite{Janyska:2010bpn}:
\begin{dfn}
Let $\mathbb R^+$ be the set of strictly positive real numbers. A \textbf{positive space} is a set $P$ equipped with two operations $+\colon P\times P\to P$ and $\mathbb R^+\times P\to P$ with the following properties:
\begin{itemize}
    \item $+$ is associative and commutative.
    \item For all $u\in P:1u=u$.
    \item For all $x,y\in\mathbb R^+$ and $u\in P:(x\cdot y)u=x(yu)$ and $(x+y)u=xu+yu$.
    \item For all $x\in\mathbb R^+$ and $u,v\in P:x(u+v)=xu+xv$.
    \item The operation $\mathbb R^+\times P\to P$ is a left free and transitive action of the group $(\mathbb R^+,\,\cdot\,)$ on $P$.
\end{itemize}
\end{dfn}
That being said, let $\mathbb T$ be the positive space associated to the units of time. $\mathbb T$ can be extended to a 1-dimensional oriented real vector space $V^1$:
\begin{dfn}
Let $P$ be a positive space, then its \textbf{extension} consists of an oriented real vector space $X^1$ and a function $i\colon P\to X$ such that
\begin{itemize}
    \item the image of $i$ is equal to the positive part of $X$.
    \item the astriction\footnote{Let $f\colon A\to B$ be some function and $f(A)\subset U\subset B$, then the obvious function $A\to U$ is called an astriction of $f$.} of $i$ onto the image is a homomorphism of positive spaces.
    \end{itemize}
    Given two extensions of $P$, there is an obvious identification of the positive parts and this bijection extends to a unique vector space isomorphism.
\end{dfn}
\hspace{15pt}Units of length are elements of a positive space $\mathbb L$ and units of area are elements of its square $\mathbb L^2$, which is defined as follows:
\begin{dfn}
Let $P$ be a positive space, then its square is a pair $(Q,(\cdot)^2)$ consisting of a positive space $Q$ and a function
\begin{equation*}
   P\ni l\mapsto l^2\in Q
\end{equation*}
such that
\begin{equation*}
    (\lambda u)^2=\lambda^2u^2
\end{equation*}
for all $u\in P$ and $\lambda>0$. Note that $(\cdot)^2\colon P\to Q$ is bijective. Thus, if $Q'$ is another square of $P$, then there is an obvious bijection $Q\to Q'$ and it actually is an isomorphism of positive spaces.
\end{dfn}
\hspace{15pt}Arrows are elements of a real vector space $V^3$ and the inner product is a function from $V^3\times V^3$ to $W^1$, the extension of $\mathbb L^2$. Since $W^1$ is oriented, it has a natural ordering and the proposition
\begin{equation*}
    \forall v\in V^3:0\leq\langle v,v\rangle
\end{equation*}
makes sense. The codomain of the associated norm is not the positive space $\mathbb L$, because the length of a vector can be equal to zero. Thus, we have to introduce a new concept: Non-negative spaces, which contain a neutral element of addition and whose elements can be multiplied by non-negative real numbers. Given the definition of positive spaces, the definition of non-negative spaces is obvious. Then the codomain of the inner product is the square root of the non-negative part of $W^1$, defined as follows:
\vspace{6pt}
\begin{dfn}
Let $X$ be a non-negative space, then its square root consists of a non-negative space $Y$ together with a function
    \begin{equation*}
        X\ni x\mapsto\sqrt{x}\in Y
    \end{equation*}
    such that
    \begin{equation*}
        \sqrt{\lambda x}=\sqrt{\lambda}\sqrt{\vphantom{\lambda}x}
    \end{equation*}
    for all $x\in X$ and $\lambda\geq 0$. Recall that two squares of the same positive space can be identified through a natural isomorphism of positive spaces. A similar construction allows us to identify two square roots of the same non-negative space through an isomorphism of non-negative spaces.
\end{dfn}
 \hspace{15pt}Lastly, we consider a speed $c$ - a homomorphism of positive spaces from $\mathbb T$ to $\mathbb L$ - and the unique inner product $V^1\times V^1\to W^1$ satisfying
\begin{equation*}
    \forall u\in\mathbb T:\sqrt{\langle u,u\rangle}=cu.
\end{equation*}
\begin{dfn}Consider some velocity $v\in L(V^1,V^3)$, then the function
\begin{align*}
    \|v\|\colon\mathbb T&\to\mathbb L\\
    u&\mapsto \|vu\|
\end{align*}
is called its \textbf{speed}.
\end{dfn}
\begin{rem}\label{rem_dimensions}
In the section on electromagnetism we consider a fixed set of units. Thus it is natural to wonder about the invariance of the physical laws under a change of units - that is, if some equation holds true for one particular choice of units, how do we know that it holds true for all possible choices of units? To answer the question, we first reformulate it within a clear mathematical setting:

In general, we consider a list of positive spaces $X_1,\ldots,X_n$ (e.g. the positive spaces associated to the base dimensions of the International System of Quantities) and a \textbf{physical quantity} with values in a real vector space $V$ is a function
\begin{equation*}
    Q\colon\prod_{i=1}^nX_i\to V
\end{equation*}
with a well-defined dimension, meaning that there exists a list
\begin{equation*}
    \alpha_1,\ldots,\alpha_n\in\mathbb Q
\end{equation*}
such that
\begin{equation*}
Q(\lambda_1x_1,\ldots,\lambda_nx_n)={\lambda_1}^{\alpha_1}\cdots{\lambda_n}^{\alpha_n} Q(x_1,\ldots,x_n)
\end{equation*} 
for all $x$ and positive real numbers $\lambda_1,\ldots,\lambda_n$.

That being said, the initial question can be rephrased as follows: If we consider two physical quantities
\begin{equation*}
    Q,Q'\colon\prod_{i=1}^nX_i\to V
\end{equation*}
then what is a sufficient condition such that the following implication holds:
\begin{equation*}
    \exists x:Q(x_1,\ldots,x_n)=Q'(x_1,\ldots,x_n)\Rightarrow \forall x:Q(x_1,\ldots,x_n)=Q'(x_1,\ldots,x_n)
\end{equation*}
A sufficient requirement that will always hold in practice is clearly that $Q$ and $Q'$ have the same dimension.
\end{rem}

%% file: 1/12-Charts.tex
\begin{dfn}
A \textbf{reference frame} on a set $M$ consists of the following data:
\begin{itemize}
    \item An affine space $A^1$ with translation space $V^1$.
    \item An affine space $A^3$ with translation space $V^3$.
    \item A function $t\colon M\to A^1$ and a function $\Pi\colon M\to A^3$ such that the induced function $F\colon M\to A^1\times A^3$ is bijective.
\end{itemize}
\end{dfn}
\begin{dfn}Suppose we have fixed a reference frame, a unit of time $e_0$ and a unit of length. Then for each choice of
\begin{itemize}
    \item an origin of time $0\in A^1$,
    \item an origin of space $O\in A^3$ and
    \item an orthonormal basis $e$ of $V^3$ (orthonormal w.r.t. the real valued inner product induced by the unit of length)
\end{itemize}
the bijective function
\begin{align*}
    A^1\times A^3&\to\mathbb R\times\mathbb R^3\\
    (t,P)&\mapsto (e^0(t-0),e(P-O))
\end{align*}
is called an \textbf{orthonormal coordinate system}.
\end{dfn}
\begin{rem}
In \cite{Mor-SR} a reference frame on $M$ is a maximal atlas $\mathcal{A}$ such that
\begin{equation}\label{AlternativeRef}
    \forall \phi,\phi'\in \mathcal{A}:\exists B\in\mathrm O(3):\dif(\phi'\circ\phi^{-1})=\begin{pmatrix}1&0\\0&B\end{pmatrix}
\end{equation}
holds true. This definition is compatible with our definition in the following sense:
Given a reference frame, a unit of time and a unit of length, then we can consider the composition of $R$ with all orthonormal coordinate systems in order to obtain such an atlas. Conversely, suppose that the following data is given:
\begin{itemize}
    \item A maximal atlas $\mathcal{A}$ satisfying (\ref{AlternativeRef}).
    \item An affine space $A^1$ with translation space $V^1$.
    \item An affine space $A^3$ with translation space $V^3$.
    \item A unit of time and a unit of length.
\end{itemize}
Let $\mathcal{O}$ be the set of all orthonormal coordinates, then we can easily construct a function $R\colon M\to A^1\times A^3$ such that
\begin{equation}\label{TwoAtlases}
    \mathcal{A}=\{\kappa\circ R:\kappa\in \mathcal{O}\}:
\end{equation}
We simply pick some $\kappa\in \mathcal{O}$
and a $\phi\in \mathcal{A}$ and set $R\coloneqq \kappa^{-1}\circ\phi$. In fact, each function $R$ satisfying (\ref{TwoAtlases}) is obviously of this form.
\end{rem}

%% file: 1/13-GM.tex
Consider a set of reference frames on a set $M$ such that $F'\circ F^{-1}$ continuously differentiable for each pair of reference frames. Note that there is a unique topology such that all reference frames are homeomorphisms. Then one way to introduce the tangent bundle is to use coordinates to define an atlas, but this is actually a detour: It is straightforward to generalize the definition of a differentiable manifold and its tangent bundle such that the reference frames form the atlas of a generalized manifold.
\begin{dfn}
Suppose that we are given a topological space $M$ and a positive integer $n$.
\begin{itemize}\label{def_GM}
    \item A generalized $n$-dimensional reference frame (an $n$-frame for short) is a pair $(A,F)$, where $A$ is an $n$-dimensional affine space and $F\colon M\to A$ is a homeomorphism.\footnote{More generally, $F$ could be a bijective function between a subset of $M$ and a subset of $A$, but this is sufficient for our purposes. Given the usual definitions of differentiable manifolds with or without boundary, a generalization should be straightforward.}
    \item Let $\mathcal{A}$ be a set of $n$-frames on $M$. If the \textbf{transition function}
    \begin{equation*}
        F'\circ F^{-1}\colon A\to A'
    \end{equation*}
    is differentiable for all $\displaystyle F,F'\in \mathcal{A}$, then $(M,\mathcal{A})$ is called a $n$-dimensional \textbf{differentiable space}.
\end{itemize}
\end{dfn}
\begin{rem}
Let $M$ be a differentiable space. We would like to emphasize that the differentials of the transition functions are not assumed to be continuous. If they are, then $M$ is called a continuously differentiable space. 
\end{rem}
\begin{dfn}
Let $M$ be an $n$-dimensional differentiable space. A \textbf{pre-tangent bundle} consists of the following data:
\begin{itemize}
    \item A set $TM$.
    \item A function $\pi\colon TM\to M$.
    \item For each $p\in M$ an $n$-dimensional real vector space structure on $T_pM\coloneqq \pi^{-1}(\{p\})$ (in particular, $T_pM$ is non-empty for all $p\in M$, i.e. $\pi$ is surjective).
\end{itemize}
\end{dfn}
\begin{dfn}
Let $(M,\mathcal{A})$ be a differentiable space. A \textbf{tangent bundle} is a pair $(TM,\mathrm{d})$, where $TM$ is a pre-tangent bundle and $\dif$ is a function defined on $\mathcal{A}$ with the following properties:
\begin{itemize}
    \item Let $(F,A)$ be some frame in $\mathcal A$ and $V$ the translation space of $A$, then $\dif F\colon TM\to A\times V$ is bijective,
    \begin{equation*}
        \forall x\in A:\forall v\in V:\dif F^{-1}(x,v)\in T_{F^{-1}(x)}M
    \end{equation*}
    and for all $p\in M$ the obvious function $\dif F_p\colon T_pM\to V$ is a vector space isomorphism. Note that we made an abuse of notation by using the same letter for a frame and the associated bijective function, i.e. $F=(A,F)$. This will happen throughout the rest of the paper.
    \item If $F$ and $F'$ are two frames in $\mathcal{A}$, then $\dif F'\circ\dif F^{-1}=\dif(F'\circ F^{-1})$.
\end{itemize}
\end{dfn}
\begin{rem}
Let $(M,\mathcal{A})$ be a continuously differentiable space.
\begin{itemize}
    \item  We can consider the unique topology on $TM$ such that $\dif F$ is a homeomorphism for all $F\in \mathcal{A}$. Then the differentials of the frames form a continuous atlas for the tangent bundle. Furthermore, each differential is a trivialization of the tangent bundle and we obtain a vector bundle.
    \item The tangent bundle is defined up to a natural isomorphism: If $(TM',\dif ')$ is a second tangent bundle, then the vector bundle isomorphism
\begin{equation*}
    \Phi\coloneqq \dif'F\circ\dif F^{-1}\colon TM\to TM'
\end{equation*}
does not depend on $F$.
\item To show the existence of a tangent bundle, we first note the existence of a pre-tangent bundle: For example, we can choose an $n$-dimensional real vector space $T_pM$ for each $p\in M$ and then consider the disjoint union. That being said, let $TM$ be a pre-tangent bundle. For each $p\in M$ we can pick a reference $F$, choose a vector space isomorphism $\dif F_p\in L(T_pM,V)$ (where $V$ is the translation space of the affine space associated to $F$) and set
\begin{equation*}
    \dif F'\coloneqq \dif(F'\circ F^{-1})_{F(p)}\circ\dif F_p
\end{equation*}
for all $F'\in \mathcal{A}$. We finally obtain a tangent bundle $(TM,\dif)$.
\end{itemize}
\end{rem}

%% file: 2/21-Galilean.tex
In view of our discussion of accelerated frames it is useful to introduce Galilean groups as a subgroup of a larger group. Furthermore, it will play an important role that the differentials of Galilean transformations are orientation-preserving (in the sense defined below), so we begin with a technical lemma:
\begin{lem}\label{lem_orientations}
Let $V^n$ and $W^n$ be two real vector spaces and suppose that $A\in L(V,W)$ is invertible. Then two bases $v_1\ldots,v_n$ and $w_1\ldots,w_n$ have the same orientation if and only if $Av_1,\ldots,Av_n$ and $Aw_1,\ldots,Aw_n$ have the same orientation. This has two implications:
\begin{itemize}
    \item The function $A$ defines a bijective function between the sets of orientations.
    \item If $V=W$, then $A$ is either \textbf{orientation-preserving} or orientation-inverting.
\end{itemize}
\end{lem}
\begin{proof}
Note that if $e\in L(V,\mathbb R^n)$ is the vector space isomorphism associated to the basis $e_1,\ldots,e_n$, then $e\circ A^{-1}\in L(V,\mathbb R^n)$ is the vector space isomorphism associated to the basis $Ae_1,\ldots,Ae_n$. That being said, suppose that $e$ and $e'$ are two bases of $V$ with the same orientation, i.e. $\det(e'\circ e^{-1})>0$. Then $e\circ A^{-1}$ and $e'\circ A^{-1}$ have the same orientation as well:
\begin{equation*}
    \det(e'\circ A^{-1}\circ(e\circ A^{-1})^{-1})=\det(e'\circ e^{-1})>0.
\end{equation*}
\end{proof}
\begin{dfn}\label{def_GKG}
    Let $F=(A^1,T,A^3,\Pi)$ and $F'=(B^1,T',B^3,\Pi')$ be two reference frames, then $F'\circ F^{-1}$ is called an element of the \textbf{general kinematic group} if and only if
\begin{itemize}
    \item $T'=\phi\circ T$, where $\phi\colon A^1\to B^1$ is affine and $\dif\phi=1$.
    \item $\forall t\in A^1$ the function
    \begin{align*}
        \Sigma_t\colon A^3&\to B^3\\
        p&\mapsto(\Pi'\circ F^{-1})(t,p)
    \end{align*}
    is affine.
    \item The image of the function
    \begin{align*}
        R\colon A^1&\to L(V^3,V^3)\\
        t&\mapsto\dif\Sigma_t
    \end{align*}
    is a subset of the rotation group (i.e. $R_t$ is orientation-preserving and orthogonal).
\end{itemize}
\end{dfn}
\begin{dfn}
    Let $F=(A^1,T,A^3,\Pi)$ and $F'=(B^1,T',B^3,\Pi')$ be two reference frames, then the transition function $T\coloneqq F'\circ F^{-1}$ is a Galilean transformation if and only if
\begin{itemize}
    \item $T$ is an element of the general kinematic group and
    \item $\forall p\in A^3$ the function \begin{align*}
        A^1&\to B^3\\
        t&\mapsto(\Pi'\circ F^{-1})(t,p)=\Sigma_t(p)
    \end{align*}
    is affine and the differential is independent of $p$.
\end{itemize}
\end{dfn}
\begin{rem}\label{remark-matrix}
Let $(V_1,V_2,W_1,W_2)$ be a list of vector spaces over the same field. Then each
\begin{equation*}
    A\in L(V_1\oplus V_2,W_1\oplus W_2)
\end{equation*}
can be identified with the unique matrix satisfying
\begin{equation*}
    \forall(v_1,v_2)\in V_1\oplus V_2:A(v_1,v_2)=\begin{pmatrix}A_{11}&A_{12}\\
    A_{21}&A_{22}
    \end{pmatrix}\begin{pmatrix}
    v_1\\v_2
    \end{pmatrix}=\begin{pmatrix}
    A_{11}v_1+A_{12}v_2\\A_{21}v_1+A_{22}v_2
    \end{pmatrix}
\end{equation*}
and the composition of two linear operators corresponds to the product of the matrices. Note that $A_{ij}\in L(V_j,W_i)$.
\end{rem}
\begin{lem}\label{Galilean_equivalent}
Let $F$ and $F'$ be two reference frames, then $T\coloneqq F'\circ F^{-1}$ is a Galilean transformation if and only if $T$ is affine and
\begin{equation*}
    \dif T=\begin{pmatrix}1&0\\v&R\end{pmatrix}
\end{equation*}
for some rotation $R\in L(V^3,V^3)$.
\end{lem}
\begin{proof}
If the transition function $F'\circ F^{-1}$ is assumed to be a Galilean transformation, then it is straightforward to prove that it has the properties listed above. To prove the other direction we first show that the function $R\colon A^1\to L(V^3,V^3)$ from definition \ref{def_GKG} is constant: Given some $v\in V^3$ we can choose $p,q\in A^3$ such that $v=q-p$ and then $R_tv=\dif(\Sigma_t)(q-p)=\Sigma_t(q)-\Sigma_t(p)$ for all $t\in A^1$. This implies that
\begin{equation*}
    A^1\ni t\mapsto R_tv\in V^3
\end{equation*}
is constant for each $v\in V^3$. That being said, let $R$ be a rotation in $L(V^3,V^3)$ for the rest of the proof.

Note that there exists a $v\in L(V^1,V^3)$ such that
\begin{equation*}
    \Sigma_{t+u}(p)=\Sigma_{t}(p)+vu
\end{equation*}
for all $t\in A^1,u\in V^1,p\in A^3$: By assumption the function $A^1\ni t\mapsto \Sigma_t(p)=:p(t)\in B^3$ is affine and its differential $v\coloneqq \dif p\in L(V^1,V^3)$ is independent of $p$.

That being said, consider $t\in A^1,u\in V^1,p\in A^3,x\in V^3$, then
\begin{equation*}
    (F'\circ F^{-1})(t+u,p+x)=(\phi(t+u),\Sigma_{t+u}(p+x))=(F'\circ F^{-1})(t,p)+(u,vu+Rx)
\end{equation*}
and this shows that $F'\circ F^{-1}$ is affine and that its differential has the desired form.
\end{proof}
\begin{rem}
By our definition the differential of a Galilean transformation is orientation-preserving. This will allow us to identify the orientations of the tangent bundle with the orientations of $V^3$ in section \ref{section_orietations}, but most importantly this implies that the transformation of vectors defined in \ref{section_transformation} is orientation-preserving.
\end{rem}
\begin{lem}\label{lem_Galilean}
Consider two frames of reference $F$ and $F'$. Fix some unit of time and length. Suppose $\phi$ and $\phi'$ are orthonormal coordinates for $F$ and $F'$. If the bases associated to $\phi$ and $\phi'$ have the same orientation, then $F'\circ F^{-1}$ is a Galilean transformation if and only if
\begin{equation*}
    \dif(\phi'\circ F'\circ F^{-1}\circ\phi^{-1})=\begin{pmatrix}1&0\\v&R\end{pmatrix}
\end{equation*}
for some $R\in\mathrm{SO}(3)$.
\end{lem}
\begin{proof}
This follows from the following facts: Suppose that $A$ and $B$ are two affine functions, then $B\circ A$ is affine and
\begin{equation*}
    \dif(B\circ A)=\dif B\circ\dif A.
\end{equation*}
Moreover, if $A$ is invertible, then $A^{-1}$ is affine and $\dif(A^{-1})=(\dif A)^{-1}$. Now the key is to realize that $\phi$ and $\phi'$ are affine and to compute the differentials.
\end{proof}

%% file: 2/22-PoinDefinition.tex
\begin{dfn}\label{Def_Lorentz}
Consider the bilinear form $\eta$ on $V^1\oplus V^3$ defined by
\begin{equation*}
    \forall v,w\in V^1:\forall x,y\in V^3:\eta\begin{pmatrix}v\\x\end{pmatrix}\begin{pmatrix}w\\y\end{pmatrix}=\langle v,w\rangle-\langle x,y\rangle.
\end{equation*}
A vector space endomorphism $\Lambda$ on $V^1\oplus V^3$ is called a \textbf{Lorentz transformation} if it preserves $\eta$, i.e.
\begin{equation*}
    \forall u,u'\in V^1\oplus V^3:\eta(\Lambda u,\Lambda u')=\eta(u,u').
\end{equation*}
\end{dfn}
\begin{rem}
Of course an endomorphism preserves $\eta$ if and only if it preserves $-\eta$. But the signature has not been chosen arbitrarily: The most important reason is explained in remark \ref{F_signature} and a more aesthetic reason is that we do not need to consider the absolute value of the metric in the definition of proper time.
\end{rem}
\begin{lem}\label{time_component}
Let $\Lambda$ be a Lorentz transformation. Note that $\Lambda_{11}\in L(V^1,V^1)$ can be identified with a real number: If $A\in L(V^1,V^1)$ and $m\colon\mathbb R\times V^1\to V^1$ is the scalar multiplication associated to $V^1$, then there is a unique $x\in\mathbb R$ such that $A=m(x,\phantom{x})$. That being said, $1\leq |\Lambda_{11}|$.
\end{lem}
\begin{proof}
Let $e_0$ be some unit of time and $(e_1,e_2,e_3)$ a basis of $V^3$ such that
\begin{equation*}
    \forall i,j:\frac{\langle e_i,e_j\rangle}{\langle e_0,e_0\rangle}=\delta_{ij}.
\end{equation*}
Then we obtain a basis $(e_0,e_1,e_2,e_3)$ of $V^1\oplus V^3$. Lastly, we define $\widehat\eta\colon V\to V^*$ through
\begin{equation*}
    \forall v,w\in V:\frac{\eta(v,w)}{\langle e_0,e_0\rangle}=:(\widehat\eta v)w.
\end{equation*}
Then
\begin{flalign*}
&1=\frac{\eta(e_0,e_0)}{\langle e_0,e_0\rangle}=\frac{\eta(\widehat\eta^{-1}e^0,\widehat\eta^{-1}e^0)}{\langle e_0,e_0\rangle}=\frac{\eta(\Lambda^{-1}\widehat\eta^{-1}e^0,\Lambda^{-1}\widehat\eta^{-1}e^0)}{\langle e_0,e_0\rangle}=\frac{\eta(\widehat\eta^{-1}e^0\Lambda,\widehat\eta^{-1}e^0\Lambda)}{\langle e_0,e_0\rangle}&\\
&=\sum_{k=0}^3\sum_{l=0}^3e^0\Lambda e_ke^0\Lambda e_l\frac{\eta(\widehat\eta^{-1}e^k,\widehat\eta^{-1}e^l)}{\langle e_0,e_0\rangle}= \Lambda^0{}_0\Lambda^0{}_0- \sum_{\alpha=1}^3 \Lambda^0{}_\alpha\Lambda^0{}_\alpha&
\end{flalign*}
and thus
\begin{equation}\label{1Plus}
    \Lambda^0{}_0\Lambda^0{}_0=1+ \sum_{\alpha=1}^3 \Lambda^0{}_\alpha\Lambda^0{}_\alpha
\end{equation}
which concludes the proof.
\end{proof}
\begin{dfn}
Let $\Lambda$ be a Lorentz transformation. Lemma \ref{time_component} shows that $\Lambda_{11}$ is either positive or negative. If $\Lambda_{11}$ is positive, then $\Lambda$ is called orthochronous. If $\Lambda$ is additionally orientation-preserving, then $\Lambda$ is called \textbf{proper orthochronous}.
\end{dfn}
\begin{dfn} Let $R$ and $R'$ be two reference frames such that $T\coloneqq R'\circ R^{-1}$ is affine. If $\dif T$ is a proper orthochronous Lorentz transformation, then $T$ is called a Poincaré transformation. We already justified the requirement of orientation-preservation in our definition of Galilean transformations.
\end{dfn}
\begin{lem}\label{lem-Poincare-equiv}
Let $F$ and $F'$ be two frames of reference. Furthermore, fix a set of natural units and let $\phi$ and $\phi'$ be orthonormal coordinates for $F$ and $F'$ such that the associated bases of $V^3$ have the same orientation.
\begin{itemize}
    \item $F'\circ F^{-1}$ is affine if and only if $\phi'\circ F'\circ F^{-1} \circ\phi^{-1}$ is affine.
    \item If $F'\circ F^{-1}$ is affine, then $\dif(F'\circ F^{-1})$ is a Lorentz transformation if and only if $\dif(\phi'\circ F'\circ F^{-1} \circ\phi^{-1})$ is a Lorentz transformation.
    \item Suppose that $F'\circ F^{-1}$ is affine and $\dif(F'\circ F^{-1})$ is a Lorentz transformation. If the bases of $V^3$ associated to $\phi$ and $\phi'$ have the same orientation, then $\dif(F'\circ F^{-1})$ is proper orthochronous if and only if $\dif(\phi'\circ F'\circ F^{-1} \circ\phi^{-1})$ is proper orthochronous.
\end{itemize}
\end{lem}
\begin{proof}
Recall the proof of lemma \ref{lem_Galilean} for the first item. To prove the second item, it helps to first introduce some new terminology:
\end{proof}
\begin{dfn}
Let $V$ be some vector space over the field $F$. If $A\colon V\times V\to F$ is bilinear, then the pair $(V,A)$ is called a \textbf{bilinear space}.
\end{dfn}
\begin{dfn}
Let $(V,A)$ and $(W,B)$ be two bilinear spaces over the same field. Then $T\in L(V,W)$ is called \textbf{product-preserving} if
\begin{equation*}
    \forall u,v\in V:B(Tu,Tv)=A(u,v).
\end{equation*}
\end{dfn}
\begin{lem}
Let $U,V,W$ be bilinear spaces over the same field. If $A\in L(U,V)$ and $B\in L(V,W)$ are product-preserving, then and $A^{-1}\in L(V,U)$ and $B\circ A\in L(U,W)$ are product-preserving as well.
\end{lem}
\begin{proof}
The proof is left as an exercise.
\end{proof}
\begin{proof}[Proof of lemma \ref{lem-Poincare-equiv}]
Now the proof is straightforward: Since coordinate systems are affine and their differentials are product-preserving (w.r.t. to the Minkowski metric), the claim follows from the last lemma:
If $\dif(F'\circ F^{-1})$ is a Lorentz transformation, then
\begin{equation*}
    \dif(\phi'\circ F'\circ F^{-1}\circ\phi^{-1})=\dif(\phi')\circ\dif\circ F'\circ F^{-1})\circ(\dif\phi)^{-1}
\end{equation*}
is a Lorentz transformation and conversely, if $\dif(\phi'\circ F'\circ F^{-1}\circ\phi^{-1})$ is a Lorentz transformation, then
\begin{equation*}
    \dif(F'\circ F^{-1})=(\dif\phi')^{-1}\circ\dif(\phi'\circ F'\circ F^{-1}\circ\phi^{-1})\circ\dif\phi
\end{equation*}
is a Lorentz transformation.

Since we have already proven the second item, the third item boils down to the following fact: Since the bases $\dif\phi$ and $\dif\phi'$ have the same orientation, the determinant of $\dif(\phi'\circ F'\circ F^{-1}\circ\phi^{-1})$ is positive if and only if $\dif (R'\circ R^{-1})$ is orientation-preserving. This can easily be verified.
\end{proof}

%% file: 2/23-Rep.tex
\begin{dfn}\label{def-boosts}
We define an inner product on $V^1\oplus V^3$ as follows:
\begin{equation*}
    \forall v,w\in V^1:\forall x,y\in V^3:\begin{pmatrix}v\\x\end{pmatrix}\boldsymbol{\cdot}\begin{pmatrix}w\\y\end{pmatrix}=\langle v,w\rangle+\langle x,y\rangle
\end{equation*}
A Lorentz transformation is called a \textbf{Lorentz boost} if it is symmetric and positive w.r.t. this inner product.
\end{dfn}
\begin{lem}\label{lem_boosts-proper}
Lorentz boosts are proper orthochronous.
\end{lem}
\begin{proof}
Consider a basis like in the proof of lemma \ref{time_component}, then a Lorentz transformation $\Lambda$ is boost (proper orthochronous) if and only if the matrix $(e^i\Lambda e_j)_{0\leq i,j\leq3}$ is a boost (proper orthochronous) and thus the claim boils down to theorem 2 in \cite{MorLorentz}.
\end{proof}
\begin{dfn}
Suppose that $v\in L(V^1,V^3)$ and $\|v\|<c$. We set
\begin{equation*}
    \gamma\coloneqq \left[1-\frac{\|v\|}{c}\frac{\|v\|}{c}\right]^{-1/2}\in [1,\infty[
\end{equation*}
and we define $J\in L(V^3,V^1)$ trough
\begin{equation*}
    \forall x\in V^3:Jx=\frac{\langle vu,x\rangle}{\langle u,u\rangle}u
\end{equation*}
where $u$ is some basis of $V^1$ (but $J$ does not depend on the choice of $u$). Lastly, let $P\in L(V^3,V^3)$ be the projection of $V^3$ onto the image of $v$. Then
\begin{equation*}
    \begin{pmatrix}
    \gamma&\gamma J\\
    \gamma v&I+(\gamma-1)P
    \end{pmatrix}=:\Lambda(v)
\end{equation*}
can be verified to be a Lorentz boost.
\end{dfn}
\begin{cor}\label{Rep_Lorentz}
Let $\Lambda$ be a proper orthochronous Lorentz transformation, then there exist a unique rotation $R\in L(V^3,V^3)$ and a unique $v\in L(V^1,V^3)$ with $\|v\|<c$ such that
    \begin{equation*}
        \Lambda=\begin{pmatrix}
    1&0\\
    0&R
    \end{pmatrix}\Lambda(v).
    \end{equation*}
    Similarly, there exist a unique rotation $R'$ and a unique $v'\in L(V^1,V^3)$ with $\|v'\|<c$ such that
    \begin{equation*}
        \Lambda(v')\begin{pmatrix}
    1&0\\
    0&R'
    \end{pmatrix}=\Lambda.
    \end{equation*}
    In fact $R=R'$ and $v'=R\circ v$.
\end{cor}
\begin{proof}
This is an immediate consequence of the following two theorems:
\end{proof}
\begin{thm}\label{Boosts-Rep}
Let $\Lambda$ be a Lorentz boost, then there is a unique $v\in L(V^1,V^3)$ such that $\|v\|<c$ and $\Lambda=\Lambda(v)$.
\end{thm}
\begin{proof}
Consider the set $X\coloneqq \{v\in L(V^1,V^3):\|v\|<c\}$ and let
\begin{equation*}
Y\subset L(V^1\oplus V^3,V^1\oplus V^3)    
\end{equation*}
be the set of all boosts. It can be verified that
\begin{equation*}
    \Lambda(v)=\begin{pmatrix}
    \gamma&\gamma J\\\gamma v&I+(\gamma-1)P
    \end{pmatrix}\in Y
\end{equation*}
for all $v\in X$, so we want to show that the function $\Lambda\colon X\to Y$ is bijective. We do so by considering a basis $e$ of $V^1\oplus V^3$ like the one in the proof of lemma \ref{time_component} and showing that $\Lambda$ is the composition of bijective functions:
\begin{itemize}
    \item Consider the bijective function
    \begin{align*}
        A\colon L(V^1,V^3)&\to\mathbb R^3\\
        v&\mapsto\sum_{i=1}^3(e^i\circ v)(e_0)
    \end{align*}
    We have $A(X)=B(0,1)$ and hence we obtain a bijection $\widetilde{A}\colon X\to B(0,1)$.
    \item The function
    \begin{align*}
        B\colon B(0,1)&\to\mathbb R^3\\
        v&\mapsto\frac{v}{\sqrt{1-v^tv}}
    \end{align*}
    is bijective (the function
    \begin{align*}
        v\colon\mathbb R^3&\to B(0,1)\\
        B&\mapsto\frac{B}{\sqrt{1+B^tB}}
    \end{align*}
    is its inverse.)
    \item Let $Z\subset\mathbb R^{4\times 4}$ be the set of all boosts, then
    \begin{align*}
        C\colon\mathbb R^3&\to Z\\
        B&\mapsto\left[
	\begin{array}{c|c}
		\gamma &  B^t  \\
\hline
	B & I+\frac{BB^t}{1+\gamma} 
	\end{array}
	\right]
    \end{align*}
    with $\gamma(B)\coloneqq \sqrt{1+B^tB}$ is a bijective function (see \cite{MorLorentz} for a proof).
    \item The function
    \begin{align*}
        D\colon\mathbb R^{4\times 4}&\to L(V^1\oplus V^3,V^1\oplus V^3)\\
        A&\mapsto e^{-1}\circ A\circ e
    \end{align*}
    is bijective and $D(Z)=Y$, so we can consider the bijection $\widetilde{D}\colon Z\to Y$.
\end{itemize}
It can be verified that $\Lambda=\widetilde{D}\circ C\circ B\circ \widetilde{A}$.
\end{proof}
It is well known that each Lorentz transformation on $\mathbb{R}^4$ can be decomposed into a boost and a spatial rotation (see \cite{sexl} for example). Furthermore it was observed in \cite{MorLorentz} that this decomposition is nothing but the polar decomposition. The advantage is that the polar decomposition theorem can just as well be applied to the Lorentz transformations from definition \ref{Def_Lorentz}:
\begin{thm}
Suppose that $A\in L(V^1\oplus V^3,V^1\oplus V^3)$ is invertible.
\begin{itemize}
    \item There exists a unique pair $(\Lambda,\Omega)$ such that: $A=\Lambda\Omega$ and
    \begin{itemize}
        \item $\Lambda$ is a symmetric and positive
        \item $\Omega$ is orthogonal
    \end{itemize}
    w.r.t. the inner product from definition \ref{def-boosts}.
    \item There exists a unique pair $(\Lambda',\Omega')$ such that: $A=\Lambda'\Omega'$ and
    \begin{itemize}
        \item $\Lambda'$ is a symmetric and positive
        \item $\Omega'$ is orthogonal
    \end{itemize}
    w.r.t. the inner product from definition \ref{def-boosts}.
    \item $\Omega=\Omega'$ and $\Lambda'=\Omega\Lambda\Omega^\dagger$
    \item Let $A$ be a Lorentz transformation, then $\Lambda$ is Lorentz transformation and hence a boost. Since Lorentz transformations form a group (a subgroup of the group of vector space automorphisms on $V^1\oplus V^3$), this means that $\Omega=\Lambda^{-1}A$ is a Lorentz transformation. Thus,
    \begin{equation*}
        \Lambda'=\Omega\Lambda\Omega^{-1}
    \end{equation*}
    is a Lorentz transformation as well.
    \item Now suppose that $A$ is a proper orthochronous Lorentz transformation. Since proper orthochronous transformations form a subgroup of the Lorentz group, the last item shows that $\Omega$ is a proper orthochronous Lorentz transformation. This together with the fact that $\Omega$ is orthogonal means that there exists a rotation $R\in L(V^3,V^3)$ such that
    \begin{equation*}
        \begin{pmatrix}
        1&0\\0&R
        \end{pmatrix}=\Omega.
    \end{equation*}
    \end{itemize}
\end{thm}
\begin{proof}
The first three items are a special case of the polar decomposition theorem in the
finite-dimensional case. See \cite{Mor-GM1} for a thorough discussion. We can proceed like in the proof of lemma \ref{lem_boosts-proper} and then our claim that $\Lambda$ is a Lorentz transformation boils down to theorem 2 in \cite{MorLorentz}.
\end{proof}

%% file: 2/24-Orientations.tex
Consider a set of reference frames on a set $M$ such that all transition functions are Poincaré transformations (Galilean transformations). Then there is a unique topology on $M$ such that all reference frames are homeomorphisms and proposition 15.9 in \cite{Lee12} tells us that there are precisely two continuous orientations of the tangent bundle. Furthermore, there is a natural bijection between the orientations of $V^3$ and the continuous orientations of $M$:

Suppse that we have chosen an orientation of $V^3$. Since $V^1$ is oriented, the orientation determines an orientation of $V^1\oplus V^3$. Furthermore, if $F$ is some reference frame, then the vector space isomorphism
\begin{equation*}
    \dif F_p\in L(T_pM,V^1\oplus V^3)
\end{equation*}
allows us to assign an orientation to $T_pM$ for all $p\in M$ (see lemma \ref{lem_orientations}). The assignment is independent of $F$ and equals one of the two continuous orientations of $M$.

%% file: 2/25-WorldLine.tex
We begin this section with a summary of the main results and highlight the differences between special relativity and Newtonian mechanics:

Consider a reference frame $F$ on a set $M$ and a subset $W$ of $M$. Our goal is to define what it means that $W$ is a world line w.r.t. $F$ such that we can prove the following result: If $W$ is a world line w.r.t. $F$ and $F'$ is another reference frame, then $W$ is also a world line w.r.t. $F'$. Of course the adequate definition will depend on the assumed relation between the reference frames. In the context of Special Relativity it is natural to require that the speed of a world line w.r.t. $F$ does not exceed the speed of light and we will prove the covariance of this requirement (i.e. the theory is compatible with the experimental data). In the simpler Galilean case this is not required.
\begin{dfn}[World lines in special relativity]\label{def_world_line}
Let $W$ be a subset of $M$, $i\colon W\to M$ the inclusion and $F$ a reference frame. Suppose $t\colon M\to A^1$ and $\Pi\colon M\to A^3$ are the two projections associated to $F$. $W$ is called a \textbf{world line} w.r.t. $F$ if
\begin{itemize}
    \item the restriction of $t\colon M\to A^1$ to $W$ is injective,
    \item its image is an interval $I\subset A^1$,
    \item the function
    \begin{equation*}
        P\coloneqq \Pi\circ i\circ t^{-1}\colon I\to A^3
    \end{equation*}
    is differentiable and $\|v\|\leq c$ with $v\coloneqq \dif P\colon I\to L(V^1,V^3)$.
\end{itemize}
\end{dfn} 
\begin{thm}\label{thm_covariance_WL}
If $W\subset M$ is a world line w.r.t. to a reference frame $F$, then $W$ is a world line w.r.t. every reference frame.
\end{thm}
\begin{proof}
Suppose that $W\subset M$ is a world line w.r.t. $F$. We want to show that $W$ is also a world line w.r.t. $F'$. The proof consists of two parts: In the first part, we show that the restriction of the projection $t'\colon M\to F^1$ to $W$ is injective and that its image is an interval. In the second part, we show that $\|\dif P'<c\|$.

Part 1: Let $t\colon W\to I\subset A^1$ be the obvious bijection. It suffices to show that $t'\circ t^{-1}$ is strictly increasing. To do so, consider the basis $e$ from the proof of lemma \ref{time_component}. It suffices to show that
\begin{equation*}
    \frac{\dif t'}{\dif t}\coloneqq \frac{\dif (t'\circ t^{-1})e_0}{e_0}>0
\end{equation*}
(the LHS is obviously independent of $e$). Firstly, we note that $t'\circ t^{-1}=\Pi\circ T\circ X$ where $X\colon I\to A^1\times A^3$ is the representation of the world line w.r.t. $F$,
\begin{equation*}
    T\coloneqq F'\circ F^{-1}\colon A^1\times A^3\to B^1\times B^3
\end{equation*}
and $\Pi\colon B^1\times B^3\to B^1$ is the obvious projection. Thus, if $\Lambda\coloneqq \dif T$, then:  
\begin{flalign*}
    &\frac{\mathrm{d}t'}{\mathrm{d}t}=\frac{\mathrm{d}(\Pi\circ T\circ X)e_0}{e_0}=(e^0\circ\Lambda\circ\mathrm{d}X)e_0=e^0\Lambda\sum\nolimits_{\alpha=0}^3e^\alpha(\mathrm{d}Xe_0)e_\alpha&\\
    &=e^0\Lambda e_0+e^0\Lambda\sum\nolimits_{\alpha=1}^3e^\alpha v(e_0)e_\alpha=e^0\Lambda e_0+\sum\nolimits_{\alpha=1}^3e^\alpha v(e_0)e^0\Lambda e_\alpha&
\end{flalign*}
Note that $0<e^0\Lambda e_0$ because $\Lambda$ is orthochronous, so it suffices to show that
\begin{equation*}
    \left|\sum\nolimits_{\alpha=1}^3e^\alpha v(e_0)e^0\Lambda e_\alpha\right|<e^0\Lambda e_0.
\end{equation*}
to conclude the proof. (\ref{1Plus}) implies that
\begin{equation*}
     \sqrt{\sum\nolimits_{i=1}^3 \displaystyle\Lambda^0{}_i\Lambda^0{}_i }<e^0\Lambda e_0
\end{equation*}
and the condition on the speed is
\begin{equation*}
    \sqrt{\sum\nolimits_{\alpha=1}^3\displaystyle{e^\alpha}ve_0}=\frac{\|ve_0\|}{u}= \frac{\|ve_0\|}{ce_0}=\frac{\|v\|}{c}<1
\end{equation*}
Now the Cauchy Schwarz inequality delivers the desired result:
\begin{flalign*}
   &\left|\sum\nolimits_{\alpha=1}^3e^\alpha v(e_0)e^0\Lambda e_\alpha\right|\leq
   \sqrt{\sum\nolimits_{\alpha=1}^3 \displaystyle{e^\alpha}v(e_0)e^\alpha v(e_0)}\sqrt{\sum\nolimits_{\alpha=1}^3\displaystyle{e^0}\Lambda e_\alpha \displaystyle{e^0}\Lambda e_\alpha}&\\
   &\leq \sqrt{\sum\nolimits_{\alpha=1}^3\displaystyle{e^0}\Lambda e_\alpha \displaystyle{e^0}\Lambda e_\alpha}<e^0\Lambda e_0&
\end{flalign*}

Part 2: The key is to realize that $\|v\|<c$ and
\begin{equation*}
    0<\frac{\langle e_0,e_0 \rangle-\langle\mathrm dPe_0,\mathrm dPe_0\rangle}{\langle e_0,e_0 \rangle}=\frac{\eta(\mathrm dXe_0,\mathrm dXe_0)}{\langle e_0,e_0 \rangle}
\end{equation*}
are equivalent. Consider the obvious function $t\colon I'\to I$, then
\begin{equation*}
    X'=T\circ X\circ t
\end{equation*}
and hence by the chain rule
\begin{equation*}
        \frac{\eta(\mathrm dX'e_0,\mathrm dX'e_0)}{\langle e_0,e_0 \rangle}=\left[\frac{\eta(\dif T\dif Xe_0, \dif T\dif Xe_0)}{\langle e_0,e_0 \rangle}\circ t\right]\frac{\dif t}{\dif t'}\frac{\dif t}{\dif t'} =\left[\frac{\eta(\dif Xe_0, \dif Xe_0)}{\langle e_0,e_0 \rangle}\circ t\right]\frac{\dif t}{\dif t'}\frac{\dif t}{\dif t'}>0.
\end{equation*}
\end{proof}
\begin{rem}
In Newtonian mechanics, we do not require that the speed of world line does not exceed the speed of light, i.e. we simply drop the last item in definition \ref{def_world_line}. Then the proof of theorem \ref{thm_covariance_WL} is similar, but much simpler: We only have to show that $\dif t'/\dif t>0$ and it follows from our definition of Galilean transformations that $\dif t'/\dif t\equiv 1$.
\end{rem}
\begin{cor}\label{cor_WW}.
Let $F$ be a reference frame on a set $M$, then each $P\in A^3$ can be identified with a constant function $P\colon A^1\to A^3$ and thus with a world line $P\subset M$ (the preimage of the graph of $P\colon A^1\to A^3$ under $F$). This holds true both in special relativity and Newtonian mechanics.
\begin{proof}
This is an immediate consequence of theorem \ref{thm_covariance_WL}.
\end{proof}
\end{cor}

%% file: 2/26-Vector_Transformation.tex
\begin{thm}\label{thm_vector-transform}Let $F$ and $F'$ be two reference frames such that $F'\circ F^{-1}$ is a Galilean transformation or a Poincaré transformation and recall that each point in $F$ corresponds to a world line by corollary \ref{cor_WW}.
\begin{enumerate}
    \item Each point in $F$ has a constant velocity in $F'$. In addition, all points have the same velocity. This velocity is called the velocity of $F$ w.r.t. $F'$ and is denoted by $v(F|F')$. This allows us to define a function
    \begin{equation*}
        \Phi\colon A^3\times A^3\to V^3
    \end{equation*}
    where $\Phi(P,Q)$ is the (time-independent) vector from $P$ to $Q$ in $F'$.
    \item If $P,Q,X,Y\in A^3$ and $Q-P=Y-X$, then
    \begin{equation*}
        \Phi(P,Q)=\Phi(X,Y)
    \end{equation*}
    and thus we can define a function $T(F\to F')\colon V^3\to V^3$.
    \item If $F'\circ F^{-1}$ is a Galilean transformation, i.e.
    \begin{equation*}
        \dif(F'\circ F^{-1})=\begin{pmatrix}
         1&0\\0&R   
        \end{pmatrix}\begin{pmatrix}
         1&0\\v&1   
        \end{pmatrix}
    \end{equation*}
    for some rotation $R\in L(V^3,V^3)$ and $v\in L(V^1,V^3)$, then
    \begin{equation*}
        T(F\to F')=R.
    \end{equation*}
    \item If $F'\circ F^{-1}$ is a Poincaré transformation, i.e.
    \begin{equation*}
        \dif(F'\circ F^{-1})=\begin{pmatrix}
         1&0\\0&R   
        \end{pmatrix}\begin{pmatrix}
    \gamma&\gamma J\\
    \gamma v&I+(\gamma-1)P
    \end{pmatrix}
    \end{equation*}
    (see corollary \ref{Rep_Lorentz}), then
    \begin{equation*}
        T(F\to F')=R+(\gamma-1)(R\circ P)-\gamma (R\circ v\circ J).
    \end{equation*}
    \item In both cases $v(F|F')=Rv$ and $v(F'|F)=-v$.
    \item The items above show that $T$ is an orientation-preserving vector space isomorphism, i.e. each basis is mapped to another basis with the same orientation. This is a consequence of the requirement that the differential of a Poincaré transformation (a Galilean transformation) is an orientation-preserving vector space isomorphism.
\end{enumerate}
\end{thm}
\begin{proof}\hfill

Notation:
\begin{itemize}
    \item Given $P\in A^3$, the function
\begin{equation*}
    A^1\ni t\mapsto (t,P)\in A^1\times A^3
\end{equation*}
will be denoted by $P$ as well.
\item Given $x\in V^3$, the function
\begin{equation*}
    V^1\ni t\mapsto (t,x)\in V^1\times V^3
\end{equation*}
will be denoted by $x$ as well.
\item $T\coloneqq F'\circ F^{-1}\colon A^1\times A^3\to B^1\times B^3$
\item $\Pi^1\colon B^1\times B^3\to B^1$ and $\Pi^3\colon B^1\times B^3\to B^3$ are the canonical projections.
\end{itemize}
1.

Suppose $P\in A^3$, then
\begin{equation*}
    P'\coloneqq (\Pi^3\circ T\circ P)\circ(\Pi^1\circ T\circ P)^{-1}\colon B^1\to B^3
\end{equation*}
is its path in $F'$. $P'$ is an affine function since the composition of affine functions is affine and the inverse of an affine function is affine. This already shows that $P$ has a constant velocity in $F'$. Now we show that each point has the same velocity:

Consider
\begin{equation*}
    A\coloneqq \begin{pmatrix}1\\0\end{pmatrix}\in L(V^1,V^1\oplus V^3),
\end{equation*}
then $\dif(P')=A$ and thus
\begin{equation*}
  \dif(P')=\underbrace{\dif\Pi^3\circ\dif T\circ A}_{=\gamma Rv}\circ(\underbrace{\dif \Pi^1\circ\dif T\circ A}_{=\gamma})^{-1}= R\circ v=:v(F|F')\in L(V^1,V^3).
\end{equation*}

2. and 4. (3. is analogous)

Choose $P,Q\in A^3$ such that $Q-P=x\in V^3$. Our goal is to prove that
\begin{equation*}
    \forall t\in B^1:Q(t)-P(t)=Rx+(\gamma-1)(R\circ P)x-\gamma (R\circ v\circ J)x.
\end{equation*}
Firstly, note that
\begin{multline*}
    (\Pi^3\circ T\circ Q)\circ(\Pi^1\circ T\circ Q)^{-1}-(\Pi^3\circ T\circ P)\circ(\Pi^1\circ T\circ P)^{-1}\\
    =(\mathrm d\Pi\circ\mathrm d T)\begin{pmatrix}(\Pi^1\circ T\circ Q)^{-1}-(\Pi^1\circ T\circ P)^{-1}\\Q-P\end{pmatrix}
\end{multline*}
We now prove
\begin{equation*}
    \forall t\in B^1:(\Pi^1\circ T\circ Q)^{-1}(t)-(\Pi^1\circ T\circ P)^{-1}(t)=-Jx
\end{equation*}
since this concludes the proof:

Choose some origin $O\in M$ and let
\begin{equation*}
    \begin{array}{cc}
        F\colon A^1\to V^1 &  \widetilde F\colon B^1\to V^1\\
        G\colon A^3\to V^3&
    \widetilde G\colon B^3\to V^3\\
    H\colon A^1\times A^3\to V^1\times V^3&
    \widetilde H\colon B^1\times B^3\to V^1\times V^3
    \end{array}
\end{equation*}
be the induced bijections. Note that
\begin{equation*}
    (\Pi^1\circ T\circ P)^{-1}=F^{-1}\circ(\underbrace{\widetilde F\circ\Pi^1\circ\widetilde H^{-1}}_{=\mathrm d\Pi^1}\circ\underbrace{\widetilde H\circ T\circ H^{-1}}_{\mathrm dT}\circ\underbrace{H\circ P\circ F^{-1}}_{=P-O})^{-1}\circ\widetilde F
\end{equation*}
and thus setting $\vec P\coloneqq P-O$ for all $P\in A^3$ yields
\begin{flalign*}
&(\Pi^1\circ T\circ Q)^{-1}-(\Pi^1\circ T\circ P)^{-1}&\\
&=F^{-1}\circ(\mathrm d\Pi^1\circ\mathrm dT\circ\vec Q)^{-1}\circ\widetilde F-F^{-1}\circ(\mathrm d\Pi^1\circ\mathrm dT\circ\vec P)^{-1}\circ\widetilde F&\\
&=(\mathrm d\Pi^1\circ\mathrm dT\circ\vec Q)^{-1}\circ\widetilde F-(\mathrm d\Pi^1\circ\mathrm dT\circ\vec P)^{-1}\circ\widetilde F.&
\end{flalign*}
Since
\begin{equation*}
    \forall x\in V^3:\forall t\in V^1:(\mathrm d\Pi^1\circ\mathrm dT \circ x)^{-1}(t)=\frac{t}{\gamma}-Jx
\end{equation*}
we finally obtain the desired result.
\end{proof}

%% file: 2/27-VR.tex
\begin{thm}\hfill
\begin{itemize}
    \item If $F$ and $F'$ measure the speed of each other, then the measured speeds are equal:
    \begin{equation*}
        \|v(F|F')\|=\|v(F'|F)\|
    \end{equation*}
    \item If an observer in $F'$ represents the direction of $v(F|F')$ by an arrow, then the arrow and $v(F'|F)$ have opposite directions from the point of view of an observer in $F$. In other words, there exists a positive real number $\alpha$ such that
    \begin{equation*}
        T(F\to F')\circ v(F'|F)=-\alpha v(F|F')\in L(V^1,V^3).
    \end{equation*}
    \item If $F'\circ F^{-1}$ is a Galilean transformation, then $\alpha=1$ and if $F'\circ F^{-1}$ is a Poincaré transformation, then
    \begin{equation*}
        \alpha=\frac{1}{\gamma}
    \end{equation*}
    (this is an occurrence of length contraction).
\end{itemize}
\end{thm}
\begin{proof}
We prove the Lorentzian case, because the Galilean case is analogous and simpler:

Theorem \ref{thm_vector-transform} tells us that $v(F|F')=Rv$ and $v(F'|F)=-v$ and therefore $\|v(F|F')\|=\|v(F'|F)\|$. Since $P\circ v=v$ and
\begin{equation*}
    v\circ J\circ v=\frac{\|v\|}{c}\frac{\|v\|}{c}v\in L(V^1,V^1)
\end{equation*}
we obtain the desired result:
\begin{flalign*}
&T(F\to F')\circ v(F'|F)=-T(F\to F')\circ v=-(R\circ v)-(\gamma-1)(R\circ P\circ v)+\gamma (R\circ v\circ J\circ v)&\\
&=-\gamma Rv+\gamma RvJv=-\gamma\left[1-\frac{\|v\|}{c}\frac{\|v\|}{c}\right]Rv=-\frac{Rv}{\gamma}=-\frac{v(F|F')}{\gamma}&
\end{flalign*}
\end{proof}
\begin{rem}
In Newtonian Mechanics, we may assume that $\dif (F'\circ F^{-1})$ is a Galilean boost for each pair of reference frames - the reason is that Galilean boosts form a group. Then the equations
\begin{align*}
    T(F\to F')\circ v(F'|F)&=-v(F|F')\\
    v(P|F')&=T(F\to F')\circ v(P|F)+v(F|F')
\end{align*}
(where $P$ is a world line) simplify to
\begin{align*}
    v(F'|F)&=-v(F|F')\\
    v(P|F')&=v(P|F)+v(F|F')
\end{align*}
since $T(F\to F')=1$ for each pair of reference frames. But there is no physical motivation for this assumption. In fact, the assumption can be misleading: We then get the impression that velocity reciprocity means that
\begin{equation*}
    \forall F,F':v(F'|F)=-v(F|F'),
\end{equation*}
but since Lorentz boosts do not form a group, it then seems like velocity reciprocity does not hold true in the context of Special Relativity.
\end{rem}

%% file: 2/28-Boosts.tex
\begin{thm}
Let $F$ and $F'$ be two reference frames on $M$ such that $F'\circ F^{-1}$ is a Galilean transformation, $\phi$ and $\phi'$ are orthonormal coordinates for $F$ and $F'$. If $e$ and $e'$ are the bases of $V^3$ associated to $\phi$ and $\phi'$, then the differential of
\begin{equation*}
    \phi'\circ F'\circ F^{-1}\circ\phi^{-1}\colon\mathbb R^4\to\mathbb R^4
\end{equation*}
is a boost if and only if $F$ observes that both bases are the same, i.e.
\begin{equation*}
    \forall i:T(F'\to F)e_i{}'=e_i.
\end{equation*}
\end{thm}
\begin{proof}
Let $A\in L(V^3,V^3)$ be the isomorphism defined by $\forall i: e_i{}'=Ae_i$, then
\begin{flalign*}
&\dif(\phi'\circ F'\circ F^{-1}\circ \phi^{-1})=\dif\phi'\circ\dif(F'\circ F^{-1})\circ\dif\phi^{-1}&\\
&=\begin{pmatrix}1&0\\0&e\circ A^{-1}\end{pmatrix}\begin{pmatrix}1&0\\R\circ v&R\end{pmatrix}\begin{pmatrix}1&0\\0&e^{-1}\end{pmatrix}&\\
&=\begin{pmatrix}1&0\\e\circ A^{-1}\circ R\circ v&e\circ A^{-1}\circ R\circ e^{-1}
\end{pmatrix}&
\end{flalign*}
and $e\circ A^{-1}\circ R\circ e^{-1}=1\Leftrightarrow A=R\Leftrightarrow A=T(F'\to F)$.
\end{proof}
\begin{rem}
The last theorem does not hold if $F'\circ F^{-1}$ is a Poincaré transformation: If $\dif (\phi'\circ R'\circ R^{-1}\circ \phi^{-1})$ happens to be a boost, the basis of $F'$ is not perceived as equal to the basis of $F$ by an observer in $F$: Set $\kappa\coloneqq \phi\circ R$, then this boils down to the fact that the vector space isomorphism
\begin{equation*}
    T(\kappa\to\kappa')\in L(\mathbb R^3,\mathbb R^3)
\end{equation*}
defined in the obvious way does not map the standard basis to the standard basis.
\end{rem}

%% file: 2/29/291-acc.tex
Let $F$ be a frame on a set $M$ and $P\subset M$ a world line w.r.t. $F$. It is natural to wonder about the existence and uniqueness of a frame $F'$ (e.g. uniqueness up to an affine transformation $T$ with
\begin{equation*}
    \dif T=\begin{pmatrix}
        1&0\\
        0&R
    \end{pmatrix}
\end{equation*}
for some rotation $R$ on $V^3$) such that
\begin{enumerate}
    \item $P$ is a world line w.r.t. $F'$, $P$ is at rest in $F'$ and
    \item all points in $F'$ are world lines w.r.t. $F$.
\end{enumerate}
We consider two simple cases:
\begin{itemize}
    \item If $P$ has a constant velocity w.r.t. $F$ and the speed of $P$ is strictly smaller than $c$, then we have at least two mathematical options: We can compose $F$ with an appropriate Galilean or a Poincaré transformation to obtain a frame that even has a uniform velocity w.r.t. $F$.
    \item Suppose that $P$ performs a uniform circular motion in $F$. We intuitively expect to find 1. a frame $F'$ such that all points in $F'$ rotate around the same axis with the same angular velocity\footnote{Note that the velocity of $F'$ w.r.t. $F$ is not bounded from above: The speed of the points goes to infinity as we move away from the rotation axis.} and 2. a frame $F''$ such that all points in $F''$ have the same velocity w.r.t. $F$ (namely the velocity of $P$). In fact we can consider the composition of $F$ with appropriate transformations in the general kinematic group to construct such frames.
\end{itemize}
In summary, the general kinematics group is a natural extension of the Galilean group which allows us to consider accelerated frames in Newtonian mechanics: Two frames can be defined to be accelerated w.r.t. each other if the transition functions are in the general kinematic group, but not in the Galilean group. However, an accelerated frame is usually meant to be accelerated w.r.t. the inertial frames, which we haven't introduced yet. Strictly speaking the rest of this chapter does only apply to Newtonian mechanics, since we lack a similar extension of the Lorentz group.

%% file: 2/29/292-trans.tex
Let $F$ and $F'$ be two reference frames on a set $M$ such that the transition functions are elements of the general kinematic group. In the following we use the notation from definition \ref{def_GKG}. We will assume that $\phi$ is the identity on $A^1$ - i.e. $A^1=B^1$ and $T=T'$. (The differential of $\phi$ is the identity on $V^1$ anyways, so the generalization - if ever necessary - is trivial.)

That being said, let $w\subset M$ be a world line w.r.t. $F$ and $P\colon I\to A^3$ the position w.r.t. $F$. We assume that $P$ is twice differentiable, i.e. the velocity $v\colon I\to L(V^1,V^3)$ and the acceleration $a\colon I\to Q(V^1,V^3)$ exist. Note that $w$ is also a world line w.r.t. $F'$ and $P'=\Sigma P$ is the position w.r.t. $F'$. We make the following two technical assumptions:
\begin{itemize}
    \item $R$ and $R^{-1}$ are both differentiable w.r.t. the operator norm.
    \item For every $O\in A^3$ the function $\Sigma O\colon A^1\to B^3$ is twice differentiable.
\end{itemize}
In this situation $P'$ turns out to be twice differentiable and we now determine the relation between $v$ and $v'$ as well as $a$ and $a'$. To do so, we consider the functions
\begin{equation*}
    \Dot{\Sigma}\colon A^1\times A^3\to L(V^1,V^3)
\end{equation*}
and
\begin{equation*}
    \Ddot{\Sigma}\colon A^1\times A^3\to Q(V^1,V^3)
\end{equation*}
defined through the requirement that $\Dot{\Sigma}O=\dif(\Sigma O)$ and $\Ddot{\Sigma}O=\dif(\Dot{\Sigma}O)$ (i.e. $\Dot{\Sigma}O$ and $\Ddot{\Sigma}O$ are the velocity and the acceleration of $O$ w.r.t. $F'$). That being said, a first application of the product rule to $P'=\Sigma P$ yields
\begin{equation}\label{eq_velocity-transform}
    v'=Rv+\Dot{\Sigma}P.
\end{equation}
For later purposes it is useful to introduce $B\coloneqq \Dot{R}R^{-1}$ and differentiating (\ref{eq_velocity-transform}) yields
\begin{equation}\label{eq_acceleration-trans}
a'=Ra+2\Dot{R}v+\Ddot{\Sigma}P=Ra+2BRv+\Ddot{\Sigma}P.
\end{equation}
The choice of an origin allows us to further decompose the right-hand side of (\ref{eq_velocity-transform}) and (\ref{eq_acceleration-trans}): Suppose that $O\in A^3$ and set $x\coloneqq  P-O$, then we have $\Sigma P=\Sigma O+Rx$ and hence by the product rule:
\begin{align}\label{eq_velocity-transform_origin}
    v'&=Rv+BRx+\Dot{\Sigma}O\\\label{eq_acc-transform_origin}
a'&=Ra+2BRv+BBRx+\Dot{B}Rx
\end{align}
We finally use the following lemma to introduce the angular velocity of $F$ w.r.t. $F'$ and to rewrite these equations in a more common form.
\begin{lem}
    Let $A^1$ be an affine space with translation space $V^1$ and
\begin{equation*}
    U\colon A^1\to L(V^3,V^3)
\end{equation*}
a function with the following properties:
\begin{itemize}
    \item The image of $U$ is a subset of the orthogonal group.
    \item $U$ and $U^{-1}$ are both differentiable.
\end{itemize}
Let $\Dot{U}$ be the differential of $U$, i.e. $\Dot{U}=\dif U\colon A^1\times V^1\to L(V^3,V^3)$, then $\Dot{U}U^{-1}$ is anti-symmetric. Thus, if we fix an orientation of $V^3$ and a unit of length, then there is a unique
\begin{equation*}
    \omega\colon A^1\times V^1\to V^3
\end{equation*}
such that
\begin{equation*}
   \Dot{U}U^{-1}v=\omega\times v
\end{equation*}
for all functions $v\colon A^1\to V^3$.
\end{lem}
\begin{proof}
Let $v$ and $w$ be two differentiable vector-valued functions on $A^1$, then
\begin{equation*}
    \langle v,w\rangle=\langle Uv,Uw\rangle
\end{equation*}
and hence by the product rule
\begin{flalign*}
    &\langle \dif v,w\rangle+\langle v,\dif w\rangle=\dif\langle v,w\rangle=\dif\langle Uv,Uw\rangle=\langle \Dot{U}v+U\dif v,Uw\rangle+\langle Uv,\Dot{U}w+U\dif w\rangle&\\
    &=\langle \Dot{U}v,Uw\rangle+\langle U\dif v,Uw\rangle+\langle Uv,\Dot{U}w\rangle+\langle Uv,U\dif w\rangle&
\end{flalign*}
Because $U$ is orthogonal, this is equivalent to
\begin{equation*}
    0=\langle \Dot{U}v,Uw\rangle+\langle Uv,\Dot{U}w\rangle.
\end{equation*}
Since $U$ is invertible and $U^{-1}\colon A^1\to L(V^3,V^3)$ is differentiable (the differential of $U^{-1}$ equals $U^{-1}\Dot{U}U^{-1}$), we can consider the differentiable functions $U^{-1}v$ and $U^{-1}w$ and we obtain
\begin{equation*}
    0=\langle \Dot{U}U^{-1}v,w\rangle+\langle v,\Dot{U}U^{-1}w\rangle.
\end{equation*}
\end{proof}
The function
\begin{equation*}
    \omega=\omega(F|F')\colon A^1\to L(V^1,V^3)
\end{equation*}
associated to $B$ through the last lemma is called the angular velocity of $F$ w.r.t. $F'$. We use it to rewrite (\ref{eq_velocity-transform_origin}) and (\ref{eq_acc-transform_origin}):
\begin{align}\label{eq_velo_trans_omega}
    v'&=Rv+\omega\times Rx+\Dot{\Sigma}O\\\label{eq_acc_trans_omega}
    a'&=Ra+2\omega\times Rv+\omega\times(\omega\times Rx)+\Dot{\omega}\times Rx+\Ddot{\Sigma}O
\end{align}

%% file: 2/29/293-inertial.tex
To define inertial frames, we fix a set of reference frames on a set $M$ such that all transition functions are elements of the general kinematic group. Since the Galilean group is a subgroup, we can introduce an equivalence relation through the definition that two frames are equivalent if and only if the transition functions are Galilean transformations.

Now the purpose of Newton's first law is to define inertial frames, i.e. a distinguished equivalence class:\\
Roughly speaking, the laws of physics discussed in Newtonian mechanics are only invariant under Galilean transformations, so the set of inertial frames can be defined to be precisely the equivalence class where these laws hold true. We use an example to illustrate the idea and to show how our formulation fits together with the original formulation of Newton's first and second law in terms of forces:

First of all, we postulate that a finite set of world lines is given.\footnote{Since the transition functions are in the general kinematic group, it makes sense to talk about world lines without referring to a reference frame} Next, we postulate the existence of a frame with the property that we can find an assignment of time-independent masses to the world lines such that the the representations of the world lines w.r.t. to the frame form a solution of the ODE known as the $n$-body problem of Newtonian mechanics. Such a frame is called inertial. Since (\ref{eq_acc_trans_omega}) reduces to $a'=Ra$ for Galilean transformations, all frames in its equivalence class are inertial as well and the masses are independent of the representative. Furthermore (\ref{eq_acc_trans_omega}) suggests that we can not find another equivalence class with inertial frames, i.e. the inertial frames form precisely one equivalence class.

If we fix a frame, then we can assign two forces to each world line: The actual force - mass times acceleration - and the force predicted by the ODE. The two forces are equal if the frame is inertial. If we interpret the forces mentioned in Newton's first and second law as the forces predicted by the ODE, then these laws are nothing but a characterization of inertial frames (and consistent with our definition):
\begin{quote}
\begin{enumerate}
    \item Every body continues in its state of rest, or of uniform motion in a straight line, unless it is compelled to change that state by forces impressed upon it.
    \item The change of motion of an object is proportional to the force impressed; and is made in the direction of the straight line in which the force is impressed.
\end{enumerate}
    
\end{quote}

%% file: 3/31-Proper_Time.tex
\begin{rem}\label{rem_def_integral}
Let $A^1$ be the affine space associated to some reference frame $R$. Since the translation space $V^1$ is oriented, $A^1$ has an obvious total order. Moreover, given $x,y\in A^1$ with $x<y$ we can consider the interval $I\coloneqq [x,y]$ and its order topology $T$. Let $\Sigma$ be the Borel $\sigma$-algebra (i.e. the smallest $\sigma$-algebra containing $T$), then there is a unique locally finite vector-valued measure $\mu\colon \Sigma\to V^1$ such that $\mu([p,q])=q-p$ for all $p,q\in I$ with $p\leq q$. If $\phi\colon I\to\mathbb R$ is continuous, then $\phi$ is bounded (because $(I,T)$ is a compact space). Hence $\phi\in \mathcal{L}^1(I,\Sigma,\mu)$ and
\begin{equation*}
    \int_x^y\phi\in V^1
\end{equation*}
is our notation for its integral.
\end{rem}
\begin{dfn}
Let $W\subset M$ be a world line, $R$ a reference frame and $t\colon M\to A^1$ the projection associated to $R$. According to our definition of world lines the image of $W$ under $t$ is an interval $I\subset A^1$ and $t\colon W\to I$ is bijective. Hence, $W$ inherits an ordering which is independent of $R$ since the differentials of Poincaré transformations are orthochronous. That being said, the \textbf{proper time} associated to a world line is the function
\begin{align*}
    W\times W&\to V^1\\
    (x,y)&\mapsto y-x
\end{align*}
defined as follows: Suppose that $x<y$ and let $e_0$ be a basis of $V^1$. Then the integral
\begin{equation*}
    y-x\coloneqq \int_{t(x)}^{t(y)}\sqrt{\frac{\dif Xe_0,\dif Xe_0}{\langle e_0,e_0\rangle}}
\end{equation*}
defined in remark \ref{rem_def_integral} is independent of $e_0$ and $R$. If $y\leq x$, then $y-x\coloneqq -(x-y)$.
\end{dfn}
\begin{proof}
To be precise, the following calculation involves two measure spaces $(I,\Sigma,\mu)$ and $(I',\Sigma',\mu')$. In addition, we make an abuse of notation by considering the obvious bijections $t\colon W\to I$ and $t\colon I'\to I$. As shown in the second part of the proof of theorem \ref{thm_covariance_WL} we have that
\begin{equation*}
    \sqrt{\frac{\displaystyle\dif X'e_0,\dif X'e_0}{\langle e_0,e_0\rangle}}=\sqrt{\frac{\dif Xe_0,\dif Xe_0}{\langle e_0,e_0\rangle}\circ t}\frac{\dif t}{\dif t'}
\end{equation*}
and hence the proof boils down to a change of variables:
\begin{flalign*}
&\int_{t'(x)}^{t'(y)}\sqrt{\frac{\displaystyle\dif X'e_0,\dif X'e_0}{\langle e_0,e_0\rangle}}=\int_{t'(x)}^{t'(y)}\sqrt{\frac{\dif Xe_0,\dif Xe_0}{\langle e_0,e_0\rangle}\circ t}\frac{\dif t}{\dif t'}=\int_{t(x)}^{t(y)}\sqrt{\frac{\dif Xe_0,\dif Xe_0}{\langle e_0,e_0\rangle}}&
\end{flalign*}
\end{proof}

%% file: 3/32-Riemann.tex
Let $F$ and $R$ be two reference frames on $M$ and consider the Lorentz transformation $\Lambda\coloneqq \dif(R\circ F^{-1})$. Furthermore, suppose that $p\in M$ and $v,w\in T_pM$. Then
\begin{equation*}
    \eta(\dif R_pv,\dif R_pw)=\eta(\Lambda\dif F_pv,\Lambda\dif F_pw)=\eta(\dif F_pv,\dif F_pw)
\end{equation*}
and hence
\begin{equation*}
    \eta_p(v,w)\coloneqq \eta(\dif F_pv,\dif F_pw)
\end{equation*}
does not depend on $F$.

%% file: 3/33-Velocity.tex
\begin{dfn}
Let $W\subset M$ be a world line, $i\colon W\to M$ the obvious inclusion and $F$ a reference frame. Note that proper time allows us to differentiate functions from $W$ to some affine space.
\begin{itemize}
    \item For all $p\in W$ the linear operator
\begin{equation*}
    U_p\coloneqq (\dif F_p)^{-1}\circ\dif(F\circ i)_p\in L(V^1,T_pM)
\end{equation*}
is called the \textbf{4-velocity} at $p$ and is clearly independent of the reference frame by our definition of the tangent bundle/by the chain rule.
\item For all $p\in W$ the quadratic function
\begin{align*}
    A_p\colon V^1&\to T_pM\\
    u&\mapsto (\dif F_p)^{-1}\dif(\dif(F\circ i)u)_pu
\end{align*}
is called the \textbf{4-acceleration} at $p$. Since all transitions functions are affine, $A_p$ is independent of $F$: If $R$ is another reference frame, then the differential of the transition function $R\circ F^{-1}$ is constant and hence
\begin{equation*}
  \dif(\dif (R\circ i)u)=\dif(\dif (R\circ F^{-1}\circ F\circ i)u)=\dif(\dif (R\circ F^{-1})\dif(F\circ i)u)=\dif (R\circ F^{-1})\dif(\dif(F\circ i)u).
\end{equation*}
\item Furthermore, if a mass $m$ is associated to $W$, then $f\coloneqq mA$ is called the \textbf{4-force}.
\end{itemize}
\end{dfn}
\begin{dfn}\label{def-classical-force}
Suppose a world line $W$, a mass $m$ and a reference frame $F$ are given. Furthermore, let $\boldsymbol{X}\colon I\to A^3$ be the trajectory w.r.t. $F$. Then its differential
\begin{equation*}
    \boldsymbol{V}\coloneqq \dif\boldsymbol{X}\colon I\to L(V^1,V^3)
\end{equation*}
is called the velocity w.r.t. $F$ and
\begin{align*}
    \gamma\coloneqq \left[1-\frac{\|\boldsymbol{V}\|}{c}\frac{\|\boldsymbol{V}\|}{c}\right]^{-1/2}\colon I\to[1,\infty[
\end{align*}
is called the Lorentz factor. Furthermore,
\begin{equation*}
    \boldsymbol{P}\coloneqq \gamma m\boldsymbol{V}\colon I\to L(V^1,V^3)
\end{equation*}
is called the momentum w.r.t. $F$ and
\begin{equation*}
    \boldsymbol{F}\coloneqq \mathrm{d}\boldsymbol{P}\colon I\to Q(V^1,V^3)
\end{equation*}
is called the force w.r.t. $F$.
\end{dfn}
\begin{lem} Suppose a world line $W$, a mass and a reference frame $F$ are given. Furthermore, let $t\colon M\to A^1$ be the projection associated to $F$. According to the definition of world lines we obtain a bijective function $t\colon W\to I$ onto some interval $I\subset A^1$. That being said, we have the following representation of the 4-velocity and the 4-force w.r.t. $F$: Let $u$ be a basis of $V^1$, then
\begin{equation}\label{eq_4velocity}
    \dif F\circ Uu\circ t^{-1}=\gamma(u,\boldsymbol{V}u)
\end{equation}
and
\begin{equation}\label{eq_4force}
    \dif F\circ fu\circ t^{-1}=\gamma(u\frac{\langle\boldsymbol{F}u,\boldsymbol{V}u\rangle}{\langle u,u\rangle},\boldsymbol{F}u)
\end{equation}
where the sections $Uu\colon W\to TM$ and $fu\colon W\to TM$ are defined in the obvious way.
\end{lem}
\begin{proof} We use the following two facts:
\begin{itemize}
    \item Set $\tau\coloneqq t^{-1}\colon I\to W$. According to our definition of proper time and the fundamental theorem of calculus we have that $\dif\tau/\dif t=1/\gamma$. Thus $\dif t/\dif\tau=\gamma\circ t$ according to the inverse function rule.
    \item Let $X\colon I\to A^1\times A^3$ be the 4-position w.r.t. $F$, then $X=F\circ i\circ\tau$.
\end{itemize}
Now the proof of (\ref{eq_4velocity}) is straightforward:
\begin{flalign*}
&\dif F\circ Uu\circ \tau=\dif F\circ (\dif F)^{-1}\circ\dif(F\circ i)u\circ \tau=\dif(F\circ i)u\circ \tau=\dif(F\circ i\circ \tau\circ t)u\circ \tau&\\
&=\dif(X\circ t)u\circ \tau=(\dif X)u\frac{(\dif t)u}{u}\circ \tau=(\dif X)u\frac{\dif t}{\dif\tau}\circ \tau=\gamma(\dif X)u=\gamma(u,\boldsymbol{V}u)&
\end{flalign*}
Next, we want to prove (\ref{eq_4force}). Set $\boldsymbol{U}\coloneqq \gamma\dif X$, then the equation above implies that
\begin{flalign*}
 &\dif F\circ mAu\circ \tau=\dif F\circ(\dif F)^{-1}\circ m\dif(\dif(F\circ i)u)u\circ \tau
=m\dif(\dif(F\circ i)u)u\circ \tau&\\
&=m\dif(\dif(F\circ i)u\circ \tau\circ t)u\circ \tau=\dif(m\boldsymbol{U}u\circ t)u\circ \tau=\gamma\dif(m\boldsymbol{U}u)u.&
\end{flalign*}
Furthermore
\begin{flalign*}
\dif(m\boldsymbol{U}u)u=\dif(\gamma mu,\gamma m\boldsymbol{V}u)u=(\dif(\gamma mu)u,\dif(\gamma m\boldsymbol{V}u)u)=(\dif(\gamma mu)u,\boldsymbol{F}u).
\end{flalign*}
Set $x\coloneqq \dif(\gamma mu)u\colon I\to V^1$, then all that remains to be shown is that
\begin{equation*}
    \gamma\frac{\langle\boldsymbol{F}u,\boldsymbol{V}u\rangle}{\langle u,u\rangle}u=x.
\end{equation*}
Note that
\begin{equation*}
    \eta(\boldsymbol{U}u,\boldsymbol{U}u)=\frac{\langle u,u\rangle-\langle\boldsymbol{V}u,\boldsymbol{V}u\rangle}{1-\frac{\langle\boldsymbol{V}u,\boldsymbol{V}u\rangle}{\langle u,u\rangle}}=\langle u,u\rangle
\end{equation*}
i.e. the function $\eta(\boldsymbol{U}u,\boldsymbol{U}u)\colon I\to W^1$ is constant. By the product rule
\begin{equation*}
    0=\eta(\dif(m\boldsymbol{U}u)u,\boldsymbol{U}u)=\langle x,\gamma u\rangle-\langle\gamma\boldsymbol{F}u,\gamma\boldsymbol{V}u\rangle
\end{equation*}
or equivalently $\langle x,u\rangle=\gamma\langle\boldsymbol{F}u,\boldsymbol{V}u\rangle$. This implies the desired result:
\begin{equation*}
    x=\frac{\langle x,u\rangle}{\langle u,u\rangle}u=\gamma\frac{\langle\boldsymbol{F}u,\boldsymbol{V}u\rangle}{\langle u,u\rangle}u
\end{equation*}
\end{proof}

%% file: 3/34-EM_Tensor.tex
We begin our reformulation of classical electromagnetism. The exposure in \cite{frankel_2011} has been an important inspiration.

From now on we assume that a set of units has been fixed and all quantities are defined w.r.t. these units.  For example, for each $p\in M$ the metric
\begin{equation*}
    \eta_p\colon T_pM\times T_pM\to W^1
\end{equation*}
can be identified with a physical quantity
\begin{equation*}
    \widehat{\eta}_p\colon\mathbb L\to L(T_pM,T_pM^*)
\end{equation*}
since each unit of length defines a unit of area and hence a basis of $W^1$. See remark \ref{rem_dimensions} for the precise definition of physical quantities and a discussion of the invariance of the theory under a change of units.
\begin{dfn}\label{Decomposition}
Let $n$ be an integer, $0<n<4$ and $p\in M$. Given a reference frame $R$ and a unit of length $u$, the vector space isomorphism
\begin{equation*}
    \mathcal{R}=\mathcal{R}_n\colon\Lambda^n(T_pM^*)\to\Lambda^{n-1}(V^*)\oplus\Lambda^n(V^*)
\end{equation*}
(with $V=V^3$) is defined as follows:
\begin{itemize}
    \item Firstly, note that there is a unique unit of time $e_0$ such that $ce_0=l$. In addition, the vector space isomorphism
    \begin{equation*}
        \dif R_p\in L(T_pM,V^1\oplus V^3)
    \end{equation*}
    allows to identify $V^1$ and $V^3$ with subspaces of $T_pM$. That being said, we simply define $e^0\in T_pM^*$ through the requirement that the restriction to $V^3$ is equal to zero.
    \item Now  consider some $\alpha\in\Lambda^k(T_pM^*)$ and let
\begin{equation*}
    i\colon\Lambda(V^*)\to\Lambda(T_pM^*)
\end{equation*}
be the canonical inclusion defined by the reference frame. Since
\begin{equation*}
    x\coloneqq e_0\imult \alpha
\end{equation*}
and
\begin{equation*}
    y\coloneqq \alpha-e^0\wedge x
\end{equation*}
are both inside the image of $i$,
\begin{equation*}
\mathcal{R}\alpha\coloneqq (i^{-1}x,i^{-1}y)
\end{equation*}
is well-defined.
\end{itemize}
\end{dfn}
\begin{rem}\label{rem_def_charge_fields}
From now on we assume that we are given the following data:
\begin{itemize}
    \item A reference frame $R$.
    \item Two real-valued and positive physical quantities\footnote{If a real-valued physical quantity is positive w.r.t. to one set of units, then it is positive for all sets of units.} $k$ and $\alpha$ with arbitrary dimensions. In particular, $k$ and $\alpha$ may be dimensionless, e.g. $k=\alpha=1$.
    \item A set of world lines $W$ with a mass and a charge associated to each world line in $W$.
\end{itemize}
We define charge through the requirement that Coulomb's law takes the form
\begin{equation*}
    \|\boldsymbol{F}\|=\frac{k}{4\pi}\frac{q}{d}\frac{q'}{d}
\end{equation*}
where $d$ is the distance between $q$ and $q'$. Note that the dimension of charge depends on the dimension of $k$. In order to introduce the electromagnetic field we make the idealized assumption that there exist two unique vector fields $\boldsymbol{E}$ and $\boldsymbol{B}$ from $M$ to $V^3$ such that
\begin{equation*}
   \boldsymbol{F}=q(\boldsymbol{E}+\frac{\alpha}{c}\boldsymbol{v}\times\boldsymbol{B})
\end{equation*}
for all world lines in $W$. (The dimensions of $\boldsymbol{E}$ and $\boldsymbol{B}$ depend on the dimensions of $k$ and $\alpha$ and are only equal if $\alpha$ is a speed.) We can prove the covariance of this assumption, i.e. if $R'$ is another reference frame, then there exist unique vector fields $\boldsymbol{E}'$ and $\boldsymbol{B'}$ such that
\begin{equation*}
   \boldsymbol{F}'=q(\boldsymbol{E}'+\frac{\alpha}{c}\boldsymbol{v}'\times\boldsymbol{B}')
\end{equation*}
for all world lines in $W$. In fact this is an immediate consequence of the following theorem:
\end{rem}
\begin{cor}[Covariance of the Lorentz force]\label{Cor-ExperimentalData}
TFAE in the situation of remark \ref{rem_def_charge_fields}:
\begin{itemize}
    \item There is a unique 2-form $F$ such that
\begin{equation*}
    f^\flat=q\frac{\alpha}{c}U\imult F
\end{equation*}
for all world lines in $W$.
\item There is a unique pair of vector fields $(\boldsymbol{E},\boldsymbol{B})$ such that
\begin{equation*}
   \boldsymbol{F}=q(\boldsymbol{E}+\frac{\alpha}{c}\boldsymbol{v}\times\boldsymbol{B})
\end{equation*}
for all world lines in $W$.
\end{itemize}
In case of existence and uniqueness,
\begin{equation*}
    F=\mathcal{R}^{-1}(-\boldsymbol{E}^\flat/\alpha,*\boldsymbol{B}^\flat).
\end{equation*}
\end{cor}
\begin{proof}
Note that
\begin{align*}
    V^3\oplus V^3&\to\Lambda^2(T_pM^*)\\
    (\boldsymbol{E},\boldsymbol{B})&\mapsto\mathcal{R}^{-1}(-\boldsymbol{E}^\flat/\alpha,*\boldsymbol{B}^\flat)
\end{align*}
is a vector space isomorphism for each $p\in M$. That being said, the following lemma completes the proof:
\end{proof}
\begin{lem}\label{LorentzForceLEM}
Consider the situation of remark \ref{rem_def_charge_fields}. If $\boldsymbol{E}$ and $\boldsymbol{B}$ are two vector fields from $M$ to $V^3$ and $F=\mathcal{R}^{-1}(-\boldsymbol{E}^\flat/\alpha,*\boldsymbol{B}^\flat)$, then we have the following equivalence for each world line in $W$:
\begin{equation*}
    \boldsymbol{F}= q(\boldsymbol{E}+\frac{\alpha}{c}\boldsymbol{v}\times\boldsymbol{B})\Leftrightarrow f^\flat=q\frac{\alpha}{c}U\imult F
\end{equation*}
\end{lem}
\begin{proof}
Set
\begin{equation*}
    (\mathscr{E},\mathscr{B})\coloneqq (-\boldsymbol{E}^\flat/\alpha,\bfhodge\boldsymbol{B}^\flat)
\end{equation*}
and consider the following proposition:
\begin{equation*}
    P\coloneqq \left[\frac{\boldsymbol{v}}{c}\imult\boldsymbol{F}^\flat=- q\dfrac{\alpha}{c}\boldsymbol{v}\imult\mathscr{E}\text{ and }\boldsymbol{F}^\flat= q\alpha(-\mathscr{E}-\frac{\boldsymbol{v}}{c}\imult\mathscr{B})\right]
\end{equation*}
We conclude the proof by showing the following equivalences (the last equivalence is obvious, since $\mathcal{R}_1$ is bijective):
\begin{equation*}
    \boldsymbol{F}= q(\boldsymbol{E}+\frac{\alpha}{c}\boldsymbol{v}\times\boldsymbol{B})\Leftrightarrow P \Leftrightarrow   \mathcal{R}_1(f^\flat)=\mathcal{R}_1(q\frac{\alpha}{c}U\imult F)\Leftrightarrow f^\flat=q\frac{\alpha}{c}U\imult F
\end{equation*}
Firstly, we prove that
\begin{equation*}
    (\boldsymbol{E}+\frac{\alpha}{c}\boldsymbol{v}\times\boldsymbol{B})^\flat=\alpha(-\mathscr{E}-\frac{\boldsymbol{v}}{c}\imult\mathscr{B})
\end{equation*}
in order two obtain the first equivalence: Let $\Omega\in\Lambda^3(V^*)$ be the volume form associated to the oriented inner product space $V^3$, then $\boldsymbol{X}\imult\Omega=\bfhodge\boldsymbol{X}^\flat$ (see exercise 2-28 in \cite{Lee18}) and hence
\begin{equation*}
   (\boldsymbol{v}\times\boldsymbol{B})^\flat=\boldsymbol{B}\imult \boldsymbol{v}\imult\Omega=-\boldsymbol{v}\imult \boldsymbol{B}\imult\Omega=-\boldsymbol{v}\imult\mathscr{B}.
\end{equation*}
The second equivalence is an immediate consequence of the following two equations:
\begin{align}\label{4forceR}
    \mathcal{R}_1(f^\flat)&=\gamma( \tfrac{\boldsymbol{v}}{c}\imult\boldsymbol{F}^\flat,-\boldsymbol{F}^\flat)\\\label{imultUFR}
        \mathcal{R}_1(U\imult F)&=\gamma(-\boldsymbol{v}\imult\mathscr{E},c\mathscr{E}+\boldsymbol{v}\imult \mathscr{B})
\end{align}
Proof of (\ref{4forceR}): Firstly, note that if $x\in\mathbb R$ and $X\coloneqq (\dif R)^{-1}(xe_0,\boldsymbol{X})$, then
\begin{equation*}
    \mathcal{R}_1(X^\flat)=(x,-\boldsymbol{X}^\flat).
\end{equation*}
Now the desired equation follows from
\begin{equation*}
    (\dif R)f=\gamma(\tfrac{\boldsymbol{F\cdot v}}{c}e_0,\boldsymbol{F}).
\end{equation*}
Proof of (\ref{imultUFR}): Consider $x\coloneqq e_0\imult F$ and $y\coloneqq F-e^0\wedge x$. We can use
\begin{equation*}
 U\imult F=U\imult (e^0\wedge x)+U\imult y=-e^0\wedge(U\imult x)+(U\imult e^0)\wedge x+U\imult y
\end{equation*}
and
\begin{equation*}
    (\dif R)U=\gamma\begin{pmatrix}
     ce_0\\\boldsymbol{v}
    \end{pmatrix}
\end{equation*}
to obtain the desired result.
\end{proof}
\begin{ax}
Consider the setting from remark \ref{rem_def_charge_fields}. Furthermore, suppose that
\begin{itemize}
    \item $\rho$ is the \textbf{charge density} w.r.t. $R$, i.e. for all measurable $V\subset A^3$ the integral of $\rho$ over $V$ yields the charge inside $V$.
    \item $\boldsymbol{J}$ is the \textbf{current density} w.r.t. $R$, i.e. for all surfaces $S$ in $A^3$ the surface integral of $\boldsymbol{J}$ over $S$ yields the current through $S$.
\end{itemize}
Then the \textbf{Maxwell equations} hold true:
\begin{equation*}
    \begin{array}{cc}
        \boldsymbol{\nabla\cdot E}=k\rho & \boldsymbol{\nabla\cdot B}=0 \\
        \boldsymbol{\nabla\times {}}\dfrac{\boldsymbol{E}}{\alpha}=-L_{e_0}\boldsymbol{B} & \boldsymbol{\nabla\times B}=\dfrac{k}{\alpha}\dfrac{\boldsymbol{J}}{c}+L_{e_0}\dfrac{\boldsymbol{E}}{\alpha}
    \end{array}
\end{equation*}
\end{ax}
\begin{rem}
The different forms of Maxwell's equations that appear in the literature are due to different choices of the quantities $k$ and $\alpha$:
\begin{table}[H]
    \centering
    \begin{tabular}{ccc}
    \toprule
    &$k$&$\alpha$\\
    \midrule
        SI & $1/\epsilon_0$ & $c$  \\
        Heaviside-Lorentz & 1 & 1  \\
         Gaussian & $4\pi$ & 1  \\
           \bottomrule
    \end{tabular}
\end{table}
A similar table can be found in \cite{Jackson}. We emphasize that the choice of $k$ and $\alpha$ has nothing to do with a choice of units. The units can still be chosen arbitrarily.
\end{rem}
\begin{thm}\label{Maxwell-covariance}
If we consider the $2$-form $F\coloneqq \mathcal{R}^{-1}(-\boldsymbol{E}^\flat/\alpha,*\boldsymbol{B}^\flat)$ (as explained in corollary \ref{Cor-ExperimentalData}, $F$ does not depend on $R$) and the vector $J\coloneqq (\dif R)^{-1}(c\rho e_0,\boldsymbol{J})$, then we have the following equivalences:
\begin{equation*}
    \dif F=0\Leftrightarrow\begin{cases}
    \boldsymbol{\nabla\times {}}\dfrac{\boldsymbol{E}}{\alpha}=-L_{e_0}\boldsymbol{B}\\
    \boldsymbol{\nabla\cdot B}=0
    \end{cases}
\end{equation*}
and
\begin{equation*}
    \hodge\dif\hodge F=\frac{k}{\alpha}\frac{J^\flat}{c}\Leftrightarrow\begin{cases}
    \boldsymbol{\nabla\cdot E}=k\rho \\
    \boldsymbol{\nabla\times B}=\dfrac{k}{\alpha}\dfrac{\boldsymbol{J}}{c}+L_{e_0}\dfrac{\boldsymbol{E}}{\alpha}
    \end{cases}
\end{equation*}
\end{thm}
\begin{proof}
We will prove this theorem after the following remark:
\end{proof}
\begin{rem}\label{F_signature}\hfill
\begin{itemize}
    \item The last theorem proves the covariance of Maxwell's equations:
If they hold for one reference frame, then they hold for all reference frames.
\item In addition, this shows that $J\coloneqq (\dif R)^{-1}(c\rho e_0,\boldsymbol{J})$ does not depend on the $R$, i.e. 4-current is indeed a 4-vector.
\item If we consider the Riemannian metric $-\eta$ and still define $F$ through corollary \ref{Cor-ExperimentalData}, then theorem \ref{Maxwell-covariance} only holds true with $F$ replaced by $-F$.
\item Throughout this section we assumed that a continuous orientation of $M$ had been fixed (or equivalently an orientation of $V^3$, see section \ref{section_orietations}). But the Maxwell equations are invariant under a change of orientation: If we consider the Maxwell equations in terms of...
\begin{itemize}
    \item ...$F$, then this follows from the fact that the composition of two Hodge stars (unlike a single Hodge star) is invariant under a change of the orientation.
    \item ...$\boldsymbol{E}$ and $\boldsymbol{B}$, then this can be seen as follows: If $\boldsymbol{B}$ is the magnetic field w.r.t. one orientation, then $-\boldsymbol{B}$ is the magnetic field w.r.t. the other orientation. Similarly, if $\boldsymbol{X}$ is some vector field and $\boldsymbol{\nabla\times X}$ is the rotation w.r.t. one orientation, then $-\boldsymbol{\nabla\times X}$ is the rotation w.r.t. the other orientation.
\end{itemize}
\end{itemize}
\end{rem}
\begin{proof}[Proof of theorem \ref{Maxwell-covariance}]
Warning: In this proof we consider two different Riemannian manifolds, the euclidean space $E^3$ (the affine space $A^3$ associated to the reference frame together with the inner product on $V^3$ w.r.t. the unit of length) and Minkowski space. We use bold symbols to avoid confusion: If $\beta$ is an exterior form on $E^3$, then $\mathbf{d}\beta$ is its exterior differential and $\boldsymbol{*}\beta$ is its Hodge dual. 

Firstly, we use the fact that $\boldsymbol{\nabla\cdot X}=\hodge\dif\hodge \boldsymbol{X}^\flat$ and $\boldsymbol{\nabla\times X}=(\hodge\dif \boldsymbol{X}^\flat)^\sharp$ for each vector field $\boldsymbol{X}$ (see exercise 2-28 in \cite{Lee18}) to rewrite Maxwell's equations:
\begin{equation*}
    \begin{array}{cc}
    \bfhodge\mathbf{d} \dfrac{\boldsymbol{E}^\flat}{\alpha}=-L_{e_0}\boldsymbol{B}^\flat &  \bfhodge\mathbf{d}\bfhodge \boldsymbol{B}^\flat=0\\
    \bfhodge\mathbf{d}\bfhodge \dfrac{\boldsymbol{E}^\flat}{\alpha}=\dfrac{k}{\alpha}\rho  & \bfhodge\mathbf{d} \boldsymbol{B}^\flat=\dfrac{k}{\alpha}\dfrac{\boldsymbol{J^\flat}}{c}+L_{e_0}\dfrac{\boldsymbol{E}^\flat}{\alpha}
\end{array}
\end{equation*}
Since $**=1$ on $\Lambda^3(V^*)$, we can simplify two equations:
\begin{equation*}
    \begin{array}{cc}
     \mathbf{d} \dfrac{\boldsymbol{E}^\flat}{\alpha}+L_{e_0} \bfhodge\boldsymbol{B}^\flat=0 &  \mathbf{d}\bfhodge \boldsymbol{B}^\flat=0\\
   \bfhodge\mathbf{d}\bfhodge \dfrac{\boldsymbol{E}^\flat}{\alpha}=\dfrac{k}{\alpha}\rho& -L_{e_0}\dfrac{\boldsymbol{E}^\flat}{\alpha}+\bfhodge\mathbf{d} \boldsymbol{B}^\flat=\dfrac{k}{\alpha}\dfrac{\boldsymbol{J^\flat}}{c}
\end{array}
\end{equation*}
Now we set
\begin{equation*}
    (\mathscr{E},\mathscr{B})\coloneqq (-\frac{\boldsymbol{E}^\flat}{\alpha},*\boldsymbol{B}^\flat)
\end{equation*}
and rewrite the equations one more time:
\begin{equation*}
    \begin{array}{cc}
     -\mathbf{d} \mathscr{E}+L_{e_0}\mathscr{B} =0&  \mathbf{d}\mathscr{B}=0\\
   -\bfhodge\mathbf{d}\bfhodge \mathscr{E}=\dfrac{k}{\alpha}\rho & -L_{e_0}\mathscr{E}-\bfhodge\mathbf{d}\bfhodge\mathscr{B}=-\dfrac{k}{\alpha}\dfrac{\boldsymbol{J^\flat}}{c}
\end{array}
\end{equation*}
Thus, it remains to be shown:
\begin{equation*}
    \mathcal{R}_3(\dif F)=0\Leftrightarrow\begin{cases}
    -\mathbf{d} \mathscr{E}+L_{e_0}\mathscr{B} =0\\
    \mathbf{d}\mathscr{B}=0
    \end{cases}
\end{equation*}
and
\begin{equation*}
    \mathcal{R}_1(\hodge\dif\hodge F)=\mathcal{R}_1\frac{J^\flat}{c}\Leftrightarrow\begin{cases}
    -\bfhodge\mathbf{d}\bfhodge \mathscr{E}=\dfrac{k}{\alpha}\rho \\
    -L_{e_0}\mathscr{E}-\bfhodge\mathbf{d}\bfhodge\mathscr{B}=-\dfrac{k}{\alpha}\dfrac{\boldsymbol{J^\flat}}{c}
    \end{cases}
\end{equation*}
Since
\begin{equation*}
    \mathcal{R}_1(J^\flat)=(c\rho,-\boldsymbol{J}^\flat)
\end{equation*}
(see the proof of lemma \ref{LorentzForceLEM}), the next lemma concludes the proof:
\end{proof}
\begin{lem}
Suppose $F$ is a 2-form on $M$ and $F=\mathcal{R}^{-1}(\mathscr{E},\mathscr{B})$, then:
\begin{align}\label{AHA}
    \mathcal{R}_3(\dif F)&=(-\mathbf{d}\mathscr{E}+L_{e_0}\mathscr{B},\mathbf{d}\mathscr{B})\\\label{Brh}
    \mathcal{R}_1(\hodge\dif\hodge F)&=(-\bfhodge\mathbf{d}\bfhodge\mathscr{E},-\bfhodge\mathbf{d}\bfhodge\mathscr{B}-L_{e_0}\mathscr{E})
\end{align}
\end{lem}
\begin{proof}
In the following, the isomorphisms $i$ and $\mathcal{R}$ from definition \ref{Decomposition} are mostly left implicit, e.g.
\begin{equation*}
    F=e^0\wedge\mathscr{E}+\mathscr{B}=(\mathscr{E,\mathscr{B}}).
\end{equation*}
We start by proving (\ref{AHA}) and then we use this result to prove to (\ref{Brh}):

Recall that
\begin{equation*}
    \dif\omega=\sum_{i=0}^3e^i\wedge L_{e_i}\omega
\end{equation*}
for each exterior form $\omega$. On the one hand, we can use
\begin{equation*}
    \forall i:L_Xe^i=L_X\dif x^i=\dif (L_Xx^i)=\dif(e^iX)
\end{equation*}
(see equation 4.21 in \cite{frankel_2011}) to obtain
\begin{flalign*}
&\dif(e^0\wedge\mathscr{E})=\sum_{i=0}^3e^i\wedge L_{e_i}(e^0\wedge\mathscr{E})&\\
&=\sum_{i=0}^3e^i\wedge(\underbrace{L_{e_i}e^0\wedge\mathscr{E}}_{=\dif(e^0e_i)\wedge\mathscr{E}=0}+e^0\wedge L_{e_i}\mathscr{E})=\sum_{i=0}^3e^i\wedge e^0\wedge L_{e_i}\mathscr{E}&\\
&=\sum_{i=1}^3e^i\wedge e^0\wedge L_{e_i}\mathscr{E}=-e^0\wedge\sum_{i=1}^3e^i\wedge L_{e_i}\mathscr{E}=-e^0\wedge\textbf{d}\mathscr{E}&
\end{flalign*}
and on the other hand
\begin{equation*}
    \dif\mathscr{B}=e^0\wedge L_{e_0}\mathscr{B}+\sum_{i=1}^3e^i\wedge L_{e_i}\mathscr{B}=e^0\wedge L_{e_0}\mathscr{B}+\mathbf{d}\mathscr{B}.
\end{equation*}
In summary,
\begin{equation*}
        \dif F=\dif(e^0\wedge\mathscr{E}+\mathscr{B})=\dif(e^0\wedge\mathscr{E})+\dif\mathscr{B}=e^0\wedge (-\mathbf{d}\mathscr{E}+L_{e_0}\mathscr{B})+\mathbf{d}\mathscr{B}
    =(-\mathbf{d}\mathscr{E}+L_{e_0}\mathscr{B},\mathbf{d}\mathscr{B}).
\end{equation*}
To prove (\ref{Brh}), we need the following lemma.
\end{proof}
\begin{lem}\label{hodgeLEM}
Let $V$ be an $n$-dimensional oriented real vector space together with a non-degenerate symmetric bilinear form $g$. If $e_1,\ldots,e_n$ is a positively oriented orthonormal basis of $V$, then
\begin{equation*}
    *e^{i_1\ldots i_k}=e^{i_1}g^{-1}e^{i_1}\cdots g^{i_k}g^{-1}e^{i_k}\epsilon_{i_1\ldots i_k i_{k+1}\ldots i_n}e^{i_{k+1}\ldots i_n}
\end{equation*}
where the RHS is not a sum: Suppose $0<k<n$ and $i_1<\ldots<i_k$, then $(i_{k+1},\ldots,i_n)$ is the unique tuple such that $i_{k+1}<\ldots<i_n$ and $(i_1,\ldots,i_n)$ is a permutation of $(1,\ldots,n)$.
\end{lem}
\begin{proof}
For a derivation of the coordinate representation of the Hodge dual based on the coordinate invariant definition, see page 168 in \cite{altland_von_delft_2019}.
\end{proof}
\begin{proof}[Proof of theorem \ref{Maxwell-covariance}]
Let $(e_i)_{1\leq i\leq 3}$ be a positively oriented orthonormal basis of $V^3$, then
\begin{equation*}
    (\dif R_p{}^{-1}e_i)_{0\leq i\leq 3}
\end{equation*}
is a positively oriented orthonormal basis of $T_pM$ and we can use (\ref{AHA}) and lemma \ref{hodgeLEM} to obtain that
\begin{equation*}
    \mathcal{R}\hodge F=\mathcal{R}\hodge\mathcal{R}^{-1}(\mathscr{E},\mathscr{B})=(\bfhodge\mathscr{B},-\bfhodge\mathscr{E}).
\end{equation*}
Then (\ref{AHA}) implies that
\begin{equation*}
    \mathcal{R}\dif\hodge F=\mathcal{R}\dif\mathcal{R}^{-1}\mathcal{R}\hodge F=\mathcal{R}\hodge\dif \mathcal{R}^{-1}(\bfhodge\mathscr{B},-\bfhodge\mathscr{E})=\mathcal{R}\hodge\dif \mathcal{R}^{-1}(\bfhodge\mathscr{B},-\bfhodge\mathscr{E})=(-\mathbf{d}\bfhodge\mathscr{B}-L_{e_0}\bfhodge\mathscr{E},-\mathbf{d}\bfhodge\mathscr{E}).
\end{equation*}
Let $\alpha$ be a 3-form on $M$ such that $\alpha=\mathcal{R}^{-1}(x,y)$, then we can use lemma \ref{hodgeLEM} one more time to show that $\mathcal{R}\hodge\alpha=(\bfhodge y,\bfhodge x)$ and this concludes the proof.
\end{proof}